%% file: minTruss-Revised.tex
\documentclass[letterpaper]{vldb}
\usepackage{times}
\usepackage{balance} 
\usepackage{amsmath,epsfig,subfigure}
\usepackage{algorithm}
\usepackage[noend]{algorithmic}
\usepackage{graphicx}
\usepackage{amssymb}
\usepackage{xspace}
\usepackage{multirow}
\usepackage{url}
\usepackage{epsfig}
\usepackage{amsfonts}
\usepackage[usenames,dvipsnames]{xcolor}
\usepackage[hide]{chato-notes} 
\usepackage[normalem]{ulem}

\newcommand{\taubar}{\bar\tau}
\newcommand{\ktruss}{$k$-truss\xspace}
\newcommand{\kmax}{\mbox {$k_{max}$}}

\newcommand{\ignore}[1]{}
\newcommand{\nop}[1]{}
\newcommand{\eat}[1]{}

\newcommand{\kw}[1]{{\ensuremath {\mathsf{#1}}}\xspace}

\newcommand{\stitle}[1]{\vspace{1ex} \noindent{\bf #1}}
\long\def\comment#1{}

\newcommand{\LL}[1]{#1} 
\newcommand{\laks}[1]{#1} 

\newcommand{\RV}[1]{#1}

\newcommand{\TobeDelete}[1]{\eat{}}

\newcommand{\JY}[1]{#1} 

\newtheorem{definition}{Definition}
\newtheorem{fact}{Fact}
\newtheorem{lemma}{Lemma}
\newtheorem{theorem}{Theorem}
\newtheorem{example}{Example}

\newtheorem{corollary}{Corollary}
\newtheorem{remark}{Remark}
\newtheorem{problem}{Problem}
\newtheorem{proposition}{Proposition}

\newcommand{\dist}{\kw{dist}}
\newcommand{\diam}{\kw{diam}}
\newcommand{\con}{\kw{connect}}

\newcommand{\ctc}{\kw{closest} \kw{truss} \kw{community}}

\newcommand{\ctcp}{\kw{CTC}-\kw{Problem}}
\newcommand{\maxctcp}{\kw{Maximal}-\kw{CTC}-\kw{Problem}}
\newcommand{\divctcp}{\kw{Diverse}-\kw{CTC}-\kw{Problem}}
\newcommand{\conctcp}{\kw{Constrainted}-\kw{CTC}-\kw{Problem}}
\newcommand{\ctckp}{\kw{CTCk}-\kw{Problem}}

\newcommand{\local}{\kw{Local}-\kw{CTC}}

\newcommand{\LCTC}{\kw{LCTC}}
\newcommand{\LAZY}{\kw{BD}}
\newcommand{\CTC}{\kw{Basic}}
\newcommand{\MDC}{\kw{MDC}}
\newcommand{\QDC}{\kw{QDC}}
\newcommand{\TRUS}{\kw{Truss}}

\newcommand{\framework}{\kw{Basic}}
\newcommand{\lazy}{\kw{Bulk}\kw{Delete}}
\newcommand{\minlazy}{\kw{Min}-\kw{Lazy}\kw{Forward}}

\newcommand{\mind}{\kw{mind}}

\newcommand{\mindist}{\kw{mindist}}
\newcommand{\trussmt}{\kw{K}-\kw{truss }\kw{Maintenance}}
\newcommand{\FRE}{\kw{FRE }}

\begin{document}

\title{Approximate Closest Community Search in Networks}

\author{
{Xin Huang$^{\dag}$, Laks V.S. Lakshmanan$^{\dag}$,  Jeffrey Xu Yu$^{\ddag}$, Hong Cheng$^{\ddag}$}%
\vspace{1.6mm}\\
\fontsize{10}{10}\selectfont\itshape
$^{\dag}$ University of British Columbia, 
$^{\ddag}$The Chinese University of Hong Kong\\
\fontsize{9}{9}\selectfont\ttfamily\upshape
\{xin0,laks\}@cs.ubc.ca, \{yu, hcheng\}@se.cuhk.edu.hk
}

\maketitle

\begin{abstract}
Recently, there has been significant interest in the study of the community search problem in social and information networks: given one or more query nodes, find densely connected communities containing the query nodes. However, most existing studies do not address the ``free rider" issue, that is, nodes far away from query nodes and irrelevant to them are included in the detected community. Some state-of-the-art models have attempted to address this issue, but not only are their formulated problems NP-hard, they do not admit any approximations without restrictive assumptions, which may not always hold in practice.

In this paper, given an undirected graph $G$ and a set of query nodes $Q$, we study community search using the $k$-truss based community model. We formulate our problem of finding a \emph{closest truss community} (CTC), as finding a connected $k$-truss subgraph with the largest $k$ that contains $Q$, and has the minimum diameter among such subgraphs. We prove this problem is NP-hard. Furthermore, it is NP-hard to approximate the problem within a factor $(2-\varepsilon)$, for any $\varepsilon >0 $. However, we develop a greedy algorithmic framework, which first finds a CTC 
containing $Q$, and then iteratively removes the furthest nodes from $Q$, from the graph. The method achieves 2-approximation to the optimal solution. To further improve the efficiency, we make use of a compact truss index and develop efficient algorithms for $k$-truss identification and maintenance as nodes get eliminated. In addition, using bulk deletion optimization and local exploration strategies, we propose two more efficient algorithms. One of them trades some approximation quality for efficiency while the other is a very efficient heuristic. Extensive experiments on 6 real-world networks show the effectiveness and efficiency of our community model and search algorithms.

\end{abstract}

\section{Introduction}
\input{intro}

\section{Problem Definition} \label{sec.problem}
\input{problem}

\section{Problem Analysis}\label{sec.hardness}
\input{character}

\section{Algorithms} \label{sec.framework}
\input{framework}

\section{Fast search algorithms}\label{sec.fast}
\input{fast}

\subsection{Local Exploration}\label{sec.local}
\input{local}

\section{Experiments} \label{sec.exp}
\input{exp}

\vspace*{-0.4cm}
\section{Related Work and Discussion} \label{sec.related}

\input{relate}

\vspace*{-0.4cm}
\section{Conclusion} \label{sec.conclu}

{In this paper, we study the \ctc search problem over a graph, given a set of query nodes, that is, find a densely connected community, in which nodes are  close to each other. 
Based on the dense subgraph definition of a $k$-truss, we formualte the CTC as a connected $k$-truss subgraph containing the query nodes with the largest $k$, and has the minimum diameter among such subgraphs.  We showed the problem is NP-hard and is NP-hard to approximate within a factor better than 2. We also matched this lower bound by developing a greedy algorithmic framework that provides a 2-approximation to the optimal solution. To support the efficient search of a CTC, we make use of a truss index and develop efficient methods of truss idenfication and maintenance. Futhermore, we improve the efficiency of greedy framework further using the bulk deletion optimization and local exploration strategies. Extensive experimental results on large real-world networks with ground-truth communities demonstrate the effectivenss and efficiency of our proposed community search model and solutions. 

It would be interesting to extend our search model and algorithms to directed graphs. Given the recent surge of interest in probabilisic graphs, an exciting question is how \ktruss generalizes to probabilistic graphs. The challenge is to develop extensions that are widely useful and tractable. Last but not the least, it would be interesting to extend the notions and techniques to networks with interactions between nodes.
}

\bibliographystyle{abbrv}
{\small
\bibliography{truss}
}

\end{document}

%% file: intro.tex
Community structures naturally exist in many real-world networks such as social, biological, collaboration, and communication networks. The task of community detection is to identify all communities in a network, which is a  fundamental and well-studied problem in the literature. Recently, several papers have studied a related but different problem called \emph{community search}, which is to find the community containing a given set of query nodes. 
The need for community search naturally arises in many real application scenarios, where one is motivated by the discovery of the communities in which given query nodes participate. Since the communities defined by different nodes in a network may be quite different, community search with query nodes opens up the prospects of user-centered and personalized search, with the potential of the answers being more meaningful to a user \cite{huang2014}. 
As just one example, in a social network, the community formed by a person's high school classmates can be significantly different from the community formed by her family members which in turn can be quite different from the one formed by her colleagues \cite{mcauley2012learning}.

\begin{figure}[t]
\small
\vskip -0.1in
\centering
\includegraphics[width=1.0\linewidth]{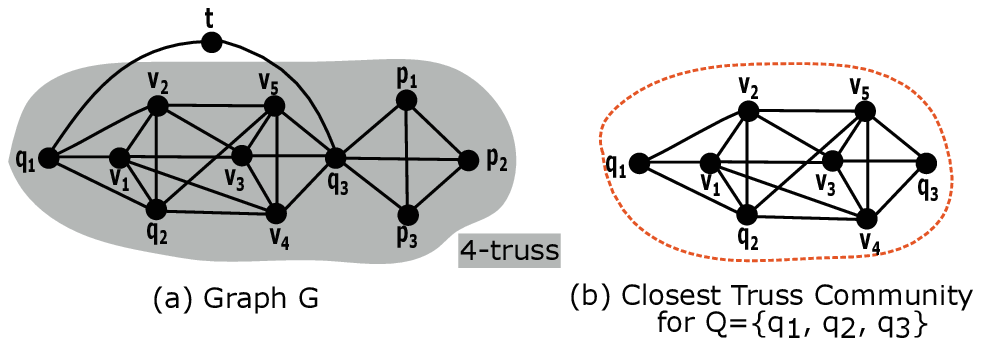}
\vskip -0.1in
\caption{Closest truss community example}\vskip -0.2in
\label{fig.community}
\end{figure}

Various community models have been proposed based on different dense subgraph structures such as $k$-core \cite{sozio2010, cui2014local, li2015influential} , $k$-truss \cite{huang2014}, quasi-clique \cite{CuiXWLW13}, weighted densest subgraph \cite{wu2015robust}, to name a few major examples. Of these, the $k$-truss as a definition of cohesive subgraph of a graph $G$, requires that each edge be contained in at least $(k-2)$ triangles within this subgraph. Consider the graph $G$ in Figure \ref{fig.community}; in the subgraph in the whole grey region (i.e., excluding the node $t$), each edge is contained in two triangles. Thus, the subgraph is a 4-truss. It is well known that most of real-world social networks are triangle-based, which always have high local clustering coefficient. Triangles are known as the fundamental building blocks of networks \cite{WangC12}.  In a social network, a triangle indicates two friends have a common friend, which shows a strong and stable relationship among three friends. Intuitively, the more common friends two people have, the stronger their relationship. In a \ktruss, each pair of friends is ``endorsed'' by at least $(k-2)$ common friends. Thus, a $k$-truss with a large value of $k$ signifies strong inner-connections between members of the subgraph. 
Huang et al.\ \cite{huang2014} proposed a community model based on the notion of \ktruss as follows. Given {\sl one} query node $q$ and a parameter $k$, a $k$-truss community containing $q$ is a maximal $k$-truss containing $q$, in which each edge is ``triangle connected'' with other edges. Triangle connectivity is strictly stronger than connectivity. The $k$-truss community model works well to find all overlapping communities containing a query node $q$. It is natural to search for communities containing a \emph{set} of query nodes in real applications, and the above community model, extended for multiple query nodes, has the following limitations. Due to the strict requirement of triangle connectivity constraint, the model may fail to discover any community for query nodes.  For example,  for the graph of Figure~\ref{fig.community}(a), and query nodes $Q=\{v_4, q_3, p_1\}$ the above  $k$-truss community model cannot find a qualified community for any $k$, since the edges $(v_4, q_3)$ and $(q_3, p_1)$ are not triangle connected in any $k$-truss. A detailed comparison of various community search models and techniques  can be found in the Section \ref{sec.related}.


In this paper, we study the problem of \emph{close community search}, i.e., given a set of query nodes, find a dense connected subgraph that contains the query nodes, in which nodes are close to each other. As a qualifying cohesive structure, we use the notion of $k$-truss for modeling a densely connected community, which inherits several good structural properties, such as $k$-edge connectivity, bounded diameter and hierarchical structure. In addition, to ensure every node included in the community is tightly related to query nodes and other nodes included in the community reported, we use  graph diameter to measure the closeness of all nodes in the community. Thus, based on $k$-truss and graph diameter, we propose a novel community model as \textbf{closest truss community (CTC)}, which requires that the all query nodes are connected in this community, the graph structure is a $k$-truss with the largest trussness $k$. In general, several such candidate communities may exist. Some of them may suffer from the so-called ``free rider effect'' formally defined and studied in \cite{wu2015robust}. While we discuss this in detail in Section~\ref{sec:fre}, we illustrate it with an example here. In Figure~\ref{fig.community}(a), for the query nodes $\{q_1,q_2,q_3\}$, the subgraph shaded grey is a 4-truss containing the query nodes. It includes the nodes $p_1, p_2, p_3$ which are intuitively not relevant to the query nodes. \laks{Specifically, they are all far away from $q_1$ and can be regarded as ``free riders''.} This 4-truss is said to suffer from the free rider effect. On the other hand, the subgraph without the nodes $\{p_1,p_2,p_3\}$ is also a 4-truss, it has the smallest diameter among all 4-trusses containing the query nodes, and does not suffer from the free rider effect. Motivated by this, we define a closest truss community as a connected \ktruss with the largest $k$ containing the query nodes and having the smallest diameter. We show that such a definition avoids the free rider effect. 
\eat{ 
\RV{Note that our model only detects one non-overlapping community for each time query.} \laks{XIN, I DON'T UNDERSTAND THE INTENT OF THE ABOVE SENTENCE. IF IT'S ONE COMMUNITY HOW CAN WE TALK ABOUT OVERLAP OR NON-OVERLAP? WHAT ARE YOU TRYING TO SAY?} 
}
A connected \ktruss with the largest $k$ containing given query nodes can be found in polynomial time. However, as we show, finding such a \ktruss with the minimum diameter is NP-hard and it is hard to approximate within a factor better than 2. Here, the approximation is w.r.t. the minimum diameter. On the other hand, we develop a greedy strategy for finding a CTC that delivers a 2-approximation to the optimal solution, thus essentially matching the lower bound. In order to make our algorithm scalable to large real networks, we propose two techniques. One of them is based on bulk deletion of nodes far away from query nodes. The second is a heuristic exploration of the local neighborhood of a Steiner tree containing the query nodes. The challenge here is that a naive application of Steiner trees may yield a \ktruss with a low value of $k$, which is undesirable. We address this challenge by developing a new notion of distances based on edge trussness. 
Specifically, we make the following contributions in this paper. 

\begin{itemize} 
\item We propose a novel community search model called \emph{closest truss community} (CTC) and motivate the problem of finding CTC containing given query nodes (Section~\ref{sec.problem}). 

\item We analyze the structural and computational properties of CTC and show that it avoids the free rider effect, is NP-hard to compute exactly or to approximate within a factor of $(2 - \varepsilon)$, for any $\varepsilon>0$  (Section~\ref{sec.hardness}). 

\item We develop a greedy 2-approximation algorithm for finding a CTC given a set of query nodes. The algorithm is based on finding, in linear time, a connected \ktruss with maximum $k$ containing the query nodes, using a simple truss index. Then successively nodes far away from the query nodes are eliminated (Section~\ref{sec.framework}). 

\item We further speed up CTC search in two ways: (1) we make use of a clever bulk deletion strategy and (2) find a Steiner tree of the query nodes and expand it into a \ktruss by exploring the local neighborhood of the Steiner tree. The first of these slightly degrades the approximation factor while the second is a heuristic (Section~\ref{sec.fast}). 

\item We extensively experiment with the various algorithms on 6 real networks. Our results show that our \ctc model can efficiently and effectively discover the queried communities on real-world networks with ground-truth communities. (Section~\ref{sec.exp}). 
\end{itemize}

\laks{In Section~\ref{sec.related}, we present a detailed comparison with related work. In Section~\ref{sec.design}, we discuss alternative candidates for community models and provide a rationale for our design decisions.} We summarize the paper in Section~\ref{sec.conclu}.

\eat{We summarize the paper in Section~\ref{sec.conclu} and discuss future work. }

%% file: problem.tex
We consider an undirected, unweighted simple graph $G=(V(G),$ $ E(G))$ with $n=|V(G)|$ vertices and $m=|E(G)|$ edges.  We denote the set of neighbors of a vertex $v$ by $N(v)$, i.e., $N(v) = \{ u\in V: (v, u)\in E \}$, and the degree of $v$ by $d(v)=|N(v)|$. We use \LL{$d_{max} = \max_{v\in V}d(v)$} to denote the maximum vertex degree in $G$. W.l.o.g we assume in this paper that the graph $G$ we consider is connected. Note that this implies that $m\geq n-1$. Table \ref{tab:notations} summarizes the frequently used notations in the paper. 

A \emph{triangle} in $G$ is a cycle of length 3.  Let $u, v, w \in V$ be the three vertices on the cycle, then we denote this triangle by $\triangle_{uvw}$.  The \emph{support} of an edge $e(u,v)\in E$ in $G$, denoted $sup_{G}(e)$, is defined as $|\{\triangle_{uvw}: w\in V\}|$. When the context is obvious, we drop the subscript and denote the support as $sup(e)$. Based on the definition of $k$-truss  \cite{cohen2008, WangC12}, we define a connected $k$-truss as follows.

\vskip -0.1in
\begin{table}[t]
\begin{center}\vspace*{-0.35cm}
\scriptsize
\caption[]{\textbf{Frequently Used Notations}}\label{tab:notations}
\begin{tabular}{|c|c|}
\hline
Notation & Description \\ \hline \hline
$G = (V(G), E(G))$ &  An undirected and connected simple graph $G$\\ \hline
$n;m$ & The number of vertices/edges in $G$ \\ \hline
$N(v)$	& The set of neighbors of $v$\\ \hline
$\sup_H(e)$	& The \emph{support} of edge $e$ in $H$ \\ \hline
$\tau(H)$	& Trussness of graph $H$  \\ \hline
$\tau(e)$	& Trussness of edge $e$ \\ \hline
$\tau(v)$	& Trussness of vertex $v$ \\ \hline
$\taubar(S)$  & The maximum trussness of connected graphs containing $S$ \\ \hline
$\diam(H)$ & The diameter of graph $H$ \\ \hline
$\dist_H(v, u)$ & The shortest distance between $v$ and $u$ in $H$ \\ \hline
$\dist_H(R, Q)$ & $\dist_H(R, Q) = \max_{v\in R, u\in Q} \dist_H(v,u)$ \\ \hline
\end{tabular}\vspace*{-0.6cm}
\end{center}
\end{table}
\vskip -0.15in

\begin{definition}
 [Connected K-Truss] Given a graph $G$ and an integer $k$, a connected $k$-truss is a connected subgraph $H \subseteq G$,  
\eat{is a connected $k$-truss subgraph of $G$, denoted as $G'\subseteq G$,}  such that $\forall e\in E(H)$, $sup_{H}(e)$ $\geq (k-2)$. 
\eat{, and $\forall v, u \in V(G')$, $v$ can reach $u$ in $G'$.} 
\end{definition} \label{def.ktruss}

Intuitively, a connected $k$-truss is a connected subgraph such that each edge $(u,v)$ in the subgraph is ``endorsed'' by $k-2$ common neighbors of $u$ and $v$ \cite{cohen2008}. 
In a connected $k$-truss graph, each node has degree at least $k-1$ and a  connected $k$-truss is also a $(k-1)$-core \cite{BatageljZ03}. Next, we define the \emph{trussness} of a subgraph, an edge, and a vertex as follows.

\begin{definition}
[Trussness] The trussness of a subgraph $H $ $\subseteq $ $ G$ is  the minimum support of an edge in $H$ plus $2$, i.e., $\tau(H) = 2+\min_{e\in E(H)}\{sup_{H}(e)\}$. The trussness of an edge $e\in E(G)$ is  $\tau(e)= \max_{H\subseteq G \wedge e\in E(H)}\{\tau(H)\}$. The trussness of a vertex $v\in V(G)$ is  $\tau(v)= \max_{H\subseteq G \wedge v \in V(H)} $ $\{\tau(H)\}$.
\end{definition}


Consider the graph $G$ in Figure~\ref{fig.community}(a). Edge $e(q_2, v_2)$ is contained in three triangles $\triangle_{q_2v_2q_1}$, $\triangle_{q_2v_2v_1}$ and $\triangle_{q_2v_2v_5}$, thus its support is $sup_{G}(e(q_2,v_2))=3$.  Suppose $H$ is the triangle $\triangle_{q_2v_2q_1}$,  then the trussness of the subgraph $H$ is  $\tau(H)=2+\min_{e\in H}{sup_{H}(e)}=3$, since each edge is contained in one triangle in $H$. The  trussness of the edge $e(q_2, v_2)$ is 4, \LL{because in the induced subgraph on vertices $\{q_1, q_2, v_1, v_2\}$, each edge is contained in two triangles in the subgraph and  any subgraph $H$ containing $e(q_2,v_2)$ has $\tau(H)\leq 4$, i.e., $\tau(e(q_2, v_2)) = $ $ \max_{H\subseteq G \wedge e\in E(H)} $ $\{\tau(H)\} $ $ =4$}. Note that the trussness of an edge $e$ of a graph $G$ could be less than $sup_G(e)+2$, e.g., $\tau(e(q_2, v_2))$ $ = 4 < 5$ $ = sup(e(q_2, v_2))+2$. 
Moreover, the vertex trussness of $q_2$ is also 4, i.e. $\tau(q_2) = 4$.

\LL{
For a set of vertices $S \subseteq V(G)$, we use $\taubar(S)$ to denote the maximum trussness of a \emph{conncted} subgraph $H$ containing $S$, i.e., $\taubar(S)= \max_{S\subseteq H\subseteq G \wedge H \text{is connected} } $ $\{\tau(H)\}$.  Notice that by definition, for $S=\emptyset$, $\taubar(\emptyset)$ is the maximum trussness of any edge in $G$. In Figure~\ref{fig.community}(a), the whole subgraph in the grey region is a 4-truss. There exists no 5-truss in $G$, and 
$\taubar(\emptyset) = 4$. We will make use of $\taubar(\emptyset)$ in Section~\ref{sec.fast}. 

\eat{For a subgrah $H\subseteq G$, we use 
$\taubar(H)$ to denote the maximum trussness of any edge in $H$. Notice that by definition, $\tau(H) \leq \taubar(H)$. In Figure~\ref{fig.community}(a), the whole subgraph in the grey region is a 4-truss. There exists no 5-truss in $G$, and 
$\taubar(G) = 4$. We will make use of $\taubar(G)$ in Section~\ref{sec.framework}
}
} 


For two nodes $u, v \in G$,  we denote by $\dist_{G}(u, v)$ the length of the shortest path between $u$ and $v$ in $G$, where $\dist_{G}(u, v) = +\infty$ if $u$ and $v$ are not connected. \laks{We make use of the notions of graph query distance and diamater in the rest of the paper.} 

\laks{ 
\begin{definition}
[Query Distance] Given a graph $G$ and a set of query nodes $Q\subset V$, for each vertex  $v\in G$, the vertex query distance of $v$ is the maximum length of a shortest path from $v$ to a query node $q\in Q$, i.e., $\dist_G(v, Q) = \max_{q\in Q} $ $ \dist_G(v, q)$. For a subgraph $H\subseteq G$, the graph query distance of $H$ is defined as $\dist_G(H, Q) = $ $\max_{u\in H} \dist_{G}(u, Q) $ $ =\max_{u\in H, q\in Q}\dist_G(u, q).$
\end{definition}  \label{def.maxdis}
} 

\begin{definition}
[Graph Diameter] The diameter of a graph $G$ is defined as the maximum length of  a shortest path in $G$, i.e., $\diam(G) =  \max_{u,v \in G}\{\dist_G(u,v)\}$. 
\end{definition}

\eat{ 
For the graph $H$ in Figure~\ref{fig.community}(b), the shortest path between $q_1$ and $q_3$ is $\langle(q_1, v_1), (v_1, v_3), (v_3, q_3)\rangle$, and $\dist_H(q_1, q_3) =3$. This is also the longest shortest path in $H$, so  $\diam(H)=3$. 
} 

\laks{
For the graph $G$ in Figure \ref{fig.community}(a) and $Q=\{q_2,$ $ q_3\}$, the vertex query distance of $v_2$ is $\dist_G(v_2, Q)= $ $\max_{q\in Q} $ $\{dist_G(v_2, q)\}$ $= 2$, since $\dist_G(v_2, q_3) = 2$ and $\dist_G(v_2, q_2) = 1$. Let $H$ be the subgraph of Figure~\ref{fig.community}(a) shaded in grey. Then query distance of $H$ is $\dist_G(H,Q) = 3$. The diameter of $H$ is $\diam(H) = 4$. 
%
} 

On the basis of the definitions of $k$-truss and graph diameter, we define the \emph{closest truss community} in a graph $G$  as follows. 


\LL{ 
\begin{definition}
[Closest Truss Community] Given a graph $G$ and a set of query nodes $Q$,  
$G'$ is a closest truss community (CTC), if $G'$ satisfies the following two conditions:
\begin{description}\label{def.community}
\vspace{-0.2cm}
  \item[(1) Connected $k$-Truss.] $G'$ is a connected $k$-truss containing $Q$ with the largest $k$, i.e., $Q\subseteq $ $G' $ $\subseteq G$ and $\forall e\in E(G')$, $sup(e) \geq k-2$;
    \vspace{-0.15cm}
  \item[(2) Smallest Diameter.] $G'$ is a subgraph of smallest diameter satisfying condition (1). That is, $\nexists G'' \subseteq G$, such that $diam(G'') $ $ < diam(G')$, and $G''$ satisfies condition (1). 
    \vspace{-0.2cm}
\end{description}
\end{definition}

Condition (1) requires that the closest community containing the query nodes $Q$ be densely connected. 
In addition, Condition (2) makes sure that each node is as close as possible to every 
\laks{other} node in the community\laks{, including the query nodes}. 
We next illustrate the notion of CTC as well as the consequence of considering Conditions (1) and (2) in different order. 
}

\begin{example}
{\em 
In Definition \ref{def.community}, we firstly consider the connected $k$-truss of $G$ containing query nodes with the largest trussness, and then among such subgraphs, regard the one with the smallest diameter as the \ctc. 
\laks{Consider the graph $G$ in Figure \ref{fig.community}(a), and $Q=\{q_1, q_2, q_3\}$; the subgraph in the region shaded grey is a 4-truss containing $Q$, and is a  subgraph with the largest trussness that contains $Q$, and has diameter 4. Notice that in Figure \ref{fig.community}(a), although the nodes $p_1, p_2, p_3$ belong to the 4-truss and are strongly connected with $q_3$, they are far away from the query node $q_1$. Figure \ref{fig.community}(b) shows another 4-truss containing $Q$ but not $p_1, p_2, p_3$, and its diameter is 3. It can be verified that this is the 4-truss with the smallest diameter. Thus, by  Condition (2) of Definition \ref{def.community}, the 4-truss graph in Figure \ref{fig.community}(a) will  not be regarded the closest truss community, whereas the one in Figure \ref{fig.community}(b) is indeed the CTC. Intuitively, the nodes $p_1, p_2, p_3$ are ``free riders'' that belong to a community defined only using Condition (1), and are avoided by Condition (2). We will see in Section \ref{sec:fre} that the definition of CTC above avoids the so-called ``free rider effect''. 
} 
\qed } 
\end{example} 

\laks{
\begin{example}
{\em Suppose we apply the conditions in Definition \ref{def.community} in the opposite order. That is, we first minimize the diameter among connected subgraphs of $G$ containing $Q$ and look for the $k$-truss subgraph with the largest $k$ among those. 
Firstly, we find that the cycle of $\{(q_1, t), $ $(t, q_3),$ $ (q_3, v_4),$ $ (v_4, q_2),$ $ (q_2, q_1)\}$ is the connected subgraph containing $Q$ with the smallest diameter 2. Then, we find that this cycle is also the $k$-truss subgraph with the largest $k$ containing itself. However, it is only a 2-truss, which has a loosely connected structure compared to Figure~\ref{fig.community}(b). This justifies the choice of the order in which Conditions (1) and (2) should be applied.   \qed } 
\end{example} 
} 


\laks{We discuss several natural candidates for community models in Section \ref{sec.design} and provide a rationale for our design decisions.} We have a choice between minimizing diameter or minimizing query distance. We address this choice in Section~\ref{sec:fre}: Example~\ref{ex-qd-dia} illustrates the value added by minimizing the diameter over minimizing just the query distance. 
The problem of \textbf{closest truss community (CTC) search} studied in this paper is stated as follows. 
\begin{problem} 
[\ctcp] Given a graph $G(V, E)$ and a set of query vertices $Q=\{v_1, ...,v_r\} \subseteq V$, find a closest truss community containing $Q$. 
\end{problem}\label{def.pro}

%% file: character.tex


\subsection{Structural Properties}
Since our closest truss community model is based on the concept of $k$-truss, the communities caputure good structral properties of $k$-truss, such as  \emph{k-edge-connected} and \emph{hierarchical structure}. In addition, since CTC is required to have minimum diameter, it also has \emph{bounded diameter}. As a result, CTC avoids the ``free rider effect''  \cite{sozio2010, wu2015robust} (see Section~\ref{sec:fre}). 

\stitle{Small diameter, k-edge-connected, hierarchical structure.} First, the diameter of a connected $k$-truss with $n$ vertices is no more than $\lfloor\frac{2n-2}{k}\rfloor$ \cite{cohen2008}. \RV{The diameter of a community is considered as an important feature of a community  \cite{Edachery99graphclustering}. } Moreover, a $k$-truss community is ($k-1$)-edge-connected \cite{cohen2008}, as it remains connected whenever fewer than $k-1$ edges are removed \cite{DBLP:books/fm/GareyJ79}. In addition, $k$-truss based community has \emph{hierarchical structure} that represents the cores of a community at different levels of granularity \cite{huang2014}, that is, $k$-truss is always contained in the $(k-1)$-truss for any $k\geq 3$. 

\stitle{Largest $k$. } We have a trivial upper bound on the maximum possible trussness of a connected $k$-truss containing the query nodes. 

\eat{Given a set of query nodes $Q$, we need to find the connected $k$-truss containing $Q$ with the largest $k$. We discuss the choice of $k$ as follows. } 

\begin{lemma}\label{lemma.largestk}
For a connected $k$-truss  $H$ satisfying definition of CTC for $Q$, we have  $k \leq \min $ $\{\tau(q_1),$ $ ..., \tau(q_r)\}$ holds.
\end{lemma}

\begin{proof}
First, we have $Q\subseteq H$. For each node $q\in Q$, $q$ cannot be contained in a $k$-truss in $G$, whenever $k >\tau(q)$. Thus, the fact that $H$ is a $k$-truss subgraph containing $Q$ implies that $k \leq$ $\min\{\tau(q_1),$ $..., \tau(q_r)\}$.
\end{proof}

\RV{
\stitle{Lower and upper bounds on diameter. } Since the distance function satisfies the triangular inequality, i.e., for all nodes $u, v, w$, $\dist_G(u,v) \leq \dist_G(u,w)+\dist_G(w, v)$, we can express the lower and upper bounds on the graph diameter in terms of the query distance as follows. 

\begin{lemma}\label{lemma.disdia}
For a graph $G(V,E)$ and a set of nodes $Q\subseteq G$, we have $\dist_G(G, Q) \leq \diam(G) \leq 2\dist_{G}(G, Q)$.
\end{lemma}

\begin{proof}
First, the diameter $\diam(G) = \max_{v, u\in G} \dist_G(v, u)$, which is clearly no less than than $\dist_G(G, Q) = \max_{v\in G, q \in Q}$ $ \dist_{G}(v, q)$ for $Q\subseteq G$. Thus, $\dist_G(G, Q) \leq \diam(G)$. 
Second, suppose that the longest shortest path in $G$ is between $v$ and $u$. Then  $\forall q\in Q$, then we have $\diam(G) = \dist(v, u)\leq \dist(v, q) + $ $\dist(q, u) \leq $ $2\dist_G(G, Q)$. The lemma follows. 
\end{proof}

}

\subsection{Free Rider Effect} 
\label{sec:fre} 



\LL{ 
In previous work on community detection, researchers \cite{sozio2010, wu2015robust} have identified an undesirable phenomenon called ``free rider effect''. Intuitively, if a definition of community admits irrelevant subgraphs in the detected community, we refer to such irrelevant subgraphs as free riders.
\RV{For instance, suppose we use the classic density definition of average internal degree $\frac{|E|}{|V|}$ as the community goodness metric. Then for a set of query nodes $Q$, the community is a subgraph containing $Q$ with the maximum density. Then, any local community for $Q$ merged with the densest subgraph part will increase the community density. However, the densest subgraph may be disconnected from or irrelvant to query nodes. This shows the simple density metric suffers from the free rider effect. Wu et al. \cite{wu2015robust} show that serveral other goodness metrics including minimum degree, local modularity, and external conductance suffer from the free rider effect.} 
\RV{Following Wu et al. \cite{wu2015robust}, we define the free rider effect as follows.} Typically, a community definition is based on a goodness metric $f(H)$ for a subgraph $H$: subgraphs with minimum\footnote{We use minimum w.l.o.g.} $f(H)$ value are defined as communities. E.g., for our CTC problem, diameter is the goodness metric: among all subgraphs with maximum trussness, the smaller the diameter of $H$, the better it is as a community. The definition of free rider effect is based on this goodness metric. We term a community query independent if it is the solution to the community search with $Q$ set to $\emptyset$. 

\eat{Furthermore, equipped with the minimum diameter constraints, our closest truss community is able to \emph{avoid the ``free-riders" effect} in community search \cite{sozio2010, wu2015robust}. We called the ``free-riders" effect as \FRE in short.   A feasible solution of \ctcp is the graph $H$, that is, $Q\subseteq H$ and $H$ is a connected $k$-truss with the largest $k$. The optimal solution of \ctcp is a feasible solution with the minimum diameter, denoted as $S$. We give the formal definition of \FRE as below.
} 

\begin{definition}
[\FRE] 
\eat{
The free-riders effect is that, for a query $Q$ with its optimal solution $S$, $\exists$ a subgraph $S^* \subseteq G$ and $S^*\nsubseteq S$, s.t., $S\cup S^*$ is also a feasible solution and $\diam(S\cup S^*) \leq \diam(S)$. $S^*$ can be irrelevant with $Q$.
} 
Given a \laks{non-empty} query $Q$, let $H$ be a solution to a community definition based on a goodness metric $f(.)$. Let $H^*$ be a (global or local) optimum solution, which is query-independent. If $f(H\cup H^*) \leq f(H)$, we say that the definition suffers from free rider effect. Here, nodes in $H^* \setminus H$ are called \emph{free riders} for the query $Q$ and community $H$. 
\end{definition}\label{def.FRE}

\begin{figure}[t]
\small
\vskip -0.1in
\centering
\includegraphics[width=0.5\linewidth]{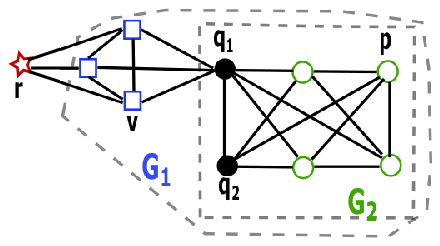}
\vskip -0.1in
\caption{A graph $G$ with $Q=\{q_1, q_2\}$.}\vskip -0.2in
\label{fig.querydistance}
\end{figure}

\laks{
\begin{example}\label{ex-qd-dia}  
{\em Consider Figure \ref{fig.querydistance}, showing a graph $G$ and query nodes $Q = \{q_1, q_2\}$. It also shows subgraphs $G_1$ and $G_2$. All three graphs -- $G, G_1$, and $G_2$ -- are 4-trusses containing $Q$. The query distance of the star node $r$ is 3, while that for all other nodes is at most 2. Thus, the query distance of $G$ is 3. The subgraph $G_1$ has the minimum query distance 2 among all 4-trusses containing $Q$. However, its diameter is 3, as the distance between square node $v$ and circle node $p$ is 3. On the other hand, the subgraph $G_2$, while having the same query distance as $G_1$, has a strictly smaller diameter 2. It has the minimum diameter among all 4-trusses containing $Q$.

Both the star node and the square nodes are free riders. The star node is the furthest from query node $q_2$ and its removal from $G$ leaves the trussness unchanged. The square nodes have the same query distance 2 as the circle node $p$. However, the square nodes are not close enough to \emph{other} nodes of the community: e.g., their distance to circle node $p$ is 3. Unlike the circle nodes, removal of square nodes leaves the trussness unchanged. Thus, the square nodes are also free riders, while the circle nodes aren't. 
\eat{
 We can see that removing the square nodes leaves the trussness unchanged and decreases the diameter, but removing any of the circle nodes \emph{decreases the trussness}. Thus both the star node and the square nodes are ``free riders''. } 
Minimizing query distance among 4-trusses eliminates the free rider star node but not the square free rider nodes, while minimizing diameter eliminates both free riders. 
\qed } 
\end{example} 
} 

\eat{ 
\RV{
Consider the graph $G$ in Figure \ref{fig.querydistance}, $G$ is a 4-truss and the query distance for star node, square nodes and circle nodes are respectively as 3, 2 and 2. The subgraph excluding star nodes, $G_1$, is a 4-truss with minimum query distance as 2; this already eliminates some free riders(the star nodes), as it is too far away from the query nodes. While adding them back to $G_1$ as $G$ does not reduce the trussness, thus if the goodness metric $f(.)$ is based only on trussness, it does fall prey to \FRE, as it would not avoid star node. The subgraph $G_2$ of $G_1$, excluding the square nodes is a 4-truss as well and has the minimum diameter. $ \diam(G_2) = 2 < 3 = \diam(G_1)$. We claim that  the square nodes are also free riders. Even though their query distances are 2, the same as the circle nodes, they are not close enough to the rest of the nodes in the community. And this phenomenon is caught by the minimum diameter criterion but not by the minimum query distance criterion.} 
}
We next show that our definition of CTC avoids the problem of free rider effect. 


\RV{ In general, there may be multiple  CTCs $H$, i.e., connected $k$-trusses with maximum trussness containing $Q$ with the minimum diameter. For example, consider the graph $G$ in Figure \ref{fig.community} and $Q=\{q_3\}$. The  subgraphs of $G$ induced respectively by $\{q_3, p_1, p_2, p_3\}$ and $\{q_3, v_3, v_4, v_5\}$  are both 4-trusses with diameter 1. Both happen to be maximal in that they are not contained in any other 4-truss with this property.}
}

\laks{
\begin{proposition}
\label{prop:fre} 
For any graph $G$ and query nodes $Q\subset V(G)$, there is a solution $H$ to the CTC search problem such that for all \RV{query-independent} optimal solutions $H^*$, either $H^*=H$, or $H\cup H^*$ is \RV{disconnected}, or $H\cup H^*$ has a strictly larger diameter than $H$. 
\end{proposition} 

\begin{proof} 
Let ${\cal C}(G,Q)$ denote the set of optimal solutions to the CTC search problem on graph $G$ and query nodes $Q$. ${\cal C}(G,Q)$ is partially ordered w.r.t. the graph containment order $\subseteq$. Let $H$ be any maximal element of ${\cal C}(G,Q)$, let $H^*$ be any \RV{query-independent} optimal solution, 
and consider $H\cup H^*$. \laks{Assume w.l.o.g. that $(H^*\setminus H) \ne \emptyset$.} Suppose that $H\cup H^*$ is a connected \ktruss with maximum trussness containing $Q$, and $\diam(H\cup H^*) \le \diam(H)$. This contradicts the maximality of $H$. 
\end{proof} 
}

\subsection{Hardness and Approximation}

\stitle{Hardness. } In the following, we show the \ctcp is NP-hard. Thereto, we define the decision version of the \ctcp. 

\vskip -0.15in
\begin{problem}
[\ctckp] Given a graph $G(V, E)$, a set of query nodes $Q=\{v_1, ...,v_r\} \subseteq V$ and parameters $k$ and $d$, test whether $G$ contains a connected $k$-truss subgraph with diameter at most $d$, that contains $Q$. 
\vspace*{-0.3cm}
\end{problem}\label{def.ctckp}
\vskip -0.15in

\begin{theorem}\label{theorem.ctckp}
The \ctckp  is \textbf{NP-hard}.
\end{theorem}

\begin{proof}
We reduce the well-known NP-hard problem of Maximum Clique (decision version) to \ctckp. Given a graph $G(V,E)$ and number $k$, the  Maximum Clique Decision problem is to check whether $G$ contains a clique of size $k$. From this, construct an instance of \ctckp, consisting of graph $G$, parameters $k$ and $d=1$, and the empty set of query nodes $Q=\emptyset$. 
\eat{
Determine whether a solution $H\subseteq G$ satifies the \ctckp conditions: i) $Q=\emptyset$ and $Q\subseteq H;$ ii) $H$ is a connected $k$-truss; and iii) $diam(H) \leq 1$. 
} 
We show that the instance of the Maximum Clique Decision problem is a YES-instance iff the corresponding instance of \ctckp is a YES-instance. 
Clearly, any clique with at least $k$ nodes is a connected $k$-truss with diameter $1$. On the other hand, given a solution $H$ for \ctckp, $H$ must contain at least $k$ nodes since $H$ is a $k$-truss, and $diam(H) = d = 1$, which implies $H$ is a clique.
\end{proof}

The hardness of \ctcp follows from this. The next natural question is whether \ctcp can be approximated. 

\eat{ 
By Theorem \ref{theorem.ctckp}, we can easily see that \ctcp is also NP-hard, because \ctckp is the decision version of \ctcp. Another way to see this is that given an instance of Maximum Clique, we can run an algorithm for \ctcp with parameter $k$ and check the diameter of the solution found. 
Then $G$ contains a $k$-clique iff the solution found has diameter $1$. 
} 

\eat{It is similar as the relation between Maximum Clique problem and Clique Decision problem. } 
\eat{
\begin{theorem}\label{theorem.ctcp}
The \ctcp is \textbf{NP-hard}.
\end{theorem}
} 



\stitle{Approximation. } For $\alpha \geq 1$, we say that an algorithm achieves an $\alpha$-approximation to the closest truss community (CTC) search problem if it outputs a connected $k$-truss subgraph $H\subseteq G$ such that $Q\subseteq H$, $\tau(H)=\tau(H^*)$ and $\diam(H)\leq \alpha\cdot\diam(H^*)$, where $H^*$ is the optimal CTC. That is, $H^*$ is a connected $k$-truss with the largest $k$ s.t. $Q\subseteq H^*$, and $diam(H^*)$ is the minimum among all such CTCs containing $Q$. Notice that the trussness of the output subgraph $H$ matches that of the optimal solution $H^*$ and that the approximation is only w.r.t. the diameter: the diameter of $H$ is required to be no more than $\alpha\cdot diam(H^*)$. 
\eat{In the following, we consistently use $H^*$ to denote the optimal CTC, i.e., $H^*$ is a connected $k$-truss with the largest $k$ s.t. $Q\subseteq H^*$, and $diam(H^*)$ is the minimum among all such CTCs containing $Q$. 
} 

\stitle{Non-Approximability. } We next prove that  \ctcp cannot be approximated within a factor better than $2$. We establish this result through a reduction, again from the Maximum Clique Decision problem to the problem of approximating \ctcp, given $k$. In the next section, we develop a $2$-approximation algorithm for \ctcp, thus essentially matching this lower bound. Notice that the \ctcp with given parameter $k$ is essentially the \ctckp. 

\laks{
\begin{theorem}
 Unless \kw{P} = \kw{NP}, for any $\varepsilon >0 $, the \ctcp with given parameter $k$ cannot be approximated in polynomial time within a factor $(2-\varepsilon)$ of the optimal.
\end{theorem}
} 

\begin{proof}

\eat{First, as discussed above, the Clique Decision problem asks, given $k$,  whether $G$ contains a clique of size $k$.} 
Suppose there exists a polynomial time algorithm $\mathbb{A}$ for the \ctcp with a given $k$ that provides a solution $H$ with an approximation factor $(2-\varepsilon)$ of the optimal solution $H^*$. Set the query nodes $Q=\emptyset$. By our assumption, we have $Q\subseteq H$, $\tau(H)=\tau(H^*) = k$ and $\diam(H)\leq (2-\varepsilon)\cdot\diam(H^*)$. Next, we use this approximation solution to exactly solve the Maximum Clique Decision problem as follows. Since the latter cannot be done in polynomial time unless \kw{P} = \kw{NP}, the theorem follows. 

Run algorithm $\mathbb{A}$ on a given instance $G$ of the Maximum Clique Decision problem, with parameter $k$ and query nodes $Q = \emptyset$. We claim that $G$ contains a clique of size $k$ iff $\mathbb{A}$ outputs a solution $H$ with $\tau(H) = k$ and $diam(H) = 1$. To see this, suppose $diam(H) = 1$, then the optimal solution $H^*$ has $\diam(H^*) \leq \diam(H) =1$, and $H^*$ is a connected $k$-truss, which shows $H^*$ is a clique of size $k$ in $G$. On the other hand, suppose $diam(H) \geq 2$. Then we have $2\cdot\diam(H^*) > (2-\varepsilon)\cdot \diam(H^*) \geq \diam(H) \geq 2$. Since diameter is an integer, we deduce that $diam(H^*) \geq 2$. In this case, $G$ cannot possibly contain a clique of size $k$, for if it did, that clique would be the optimal solution to the \ctcp on $G$, with parameter $k$, whose diameter is $1$, which contradicts the optimality of $H^*$. Thus, using algorithm $\mathbb{A}$, we can distinguish between the YES and NO instances of the Maximum Clique Decision problem. This was to be shown. 
\eat{due to the interger of $diam(H)$, if $diam(H) \geq 2$, and $\varepsilon >0 $,  we have $2\cdot\diam(H^*) > (2-\varepsilon)\cdot \diam(H^*) \geq \diam(H) \geq 2$, . Thus, we obtain $H^*$ is not a clique for $\diam(H^*) >1$. Then, we can conclude that $G$ doesn't contain a clique of size $k$; Otherwise, this clique of size $k$ can be the optimal solution $H^*$ with diameter of 1 for the \ctcp with parameter $k$, which is a contradiction. As a result, the Clique Decision problem can be solved in polynomial time by this method, which is impossible unless \kw{P} = \kw{NP}. }  

\end{proof}

\eat{ 
Based on above discussion, we can easily obtain the following theorem. 
\begin{theorem}
For any $\varepsilon >0 $, $(2-\varepsilon)$-approximation of \ctcp is \textbf{NP-hard}.
\end{theorem}
}

%% file: framework.tex
In this section, we present a greedy algorithm called \framework for the CTC search problem. Then, we show that this algorithm achieves a 2-approximation to the optimal result. Finally, we discuss procedures for an efficient implementation of the algorithm and analyze its time and space complexity.

\subsection{Basic Algorithmic Framework}

Here is an overview of our algorithm \framework. First, given a graph $G$ and query nodes $Q$, we find a maximal connected \ktruss, denoted as $G_0$, containing $Q$ and having the largest trussness.\eat{\RV{We can apply the truss decomposition algorithm \cite{cohen2008,WangC12} to find $G_0$ as follows. We delete edges from $G$ in ascending order of edge support, which increases the trussness of the remaining graph accordingly, until $Q$ gets disconnected. Then, we can obtain the largest trussness $k$, and recover $G_0$ by keeping all $k$-truss edges.}} As $G_0$ may have a large diameter, we iteratively remove nodes far away from the query nodes, while maintaining the trussness of the remainder graph at $k$.

\stitle{Algorithm.} Algorithm \ref{algo:mindia} outlines a framework for finding a closest truss community based on a greedy strategy. For query nodes $Q$, we first find a maximal connected $k$-truss $G_0$ that contains $Q$, s.t.\ $k=\tau(G_0)$ is the largest (line 1). Then, we set $l = 0$. For all $u\in G_l$ and $q\in Q$, we compute the shortest distance between $u$ and $q$ (line 4), and obtain the vertex query distance $\dist_{G_l}(u, Q)$. Among all vertices, we pick up a vertex $u^*$ with the maximum $\dist_{G_l}(u^*, Q)$, which is also the graph query distance $\dist_{G_l}(G_l, Q)$ (lines 5-6). Next, we remove the vertex $u^*$ and its incident edges from $G_l$, and delete any nodes and edges needed to restore the \ktruss property of $G_l$ (lines 7-8). We assign the updated graph as a new $G_{l}$.  Then, we repeat the above steps until  $G_{l}$ does not have a connected subgraph containing $Q$ (lines 3-9). Finally, we terminate by outputting graph $R$ as the closest truss community, where $R$ is any graph $G' \in \{G_0, ..., G_{l-1}\}$ with the smallest graph query distance $\dist_{G'}(G', Q)$ (line 10). Note that each intermediate graph $G' \in \{G_0, ..., G_{l-1}\}$ is a $k$-truss with the maximum trussness as required.


\LL{
\begin{example}
We apply Algorithm \ref{algo:mindia} on $G$ in Figure \ref{fig.community} for $Q=\{q_1,$ $ q_2, $ $ q_3\}$. First, we obtain the 4-truss subgraph $G_0$ shaded in grey, using a procedure we will shortly explain. Then, we compute all shortest distances, and get the maximum vertex query distance as  $\dist_{G_0}(p_1, Q) =4$, and $u^*= p_1$. We delete node $p_1$ and its incident edges from $G_0$; we also delete $p_2$ and $p_3$, in order to restore the 4-truss property. The resulting subgraph is $G_1$. Any further deletion of a node in the next iteration of the while loop will induce a series of deletions in line 8, eventually making the graph disconnected or containing just a part of query nodes.
As a result, the output graph $R$, shown in Figure~\ref{fig.community}(b), is just $G_1$. Also $\dist_{R}(R, Q)=3$, and $R$ happens to be the exact CTC with diameter $3$, which is optimal.
\end{example}
}


\begin{algorithm}[t]
\small
\caption{\framework($G$, $Q$)} \label{algo:mindia}
\textbf{Input:} A graph $G=(V, E)$, a set of query nodes $Q=\{q_1, ..., q_r\}$.\\
\textbf{Output:} A connected $k$-truss $R$ with a small diameter.\\
\
\begin{algorithmic}[1]

\STATE  Find a maximal connected $k$-truss containing $Q$ with the largest $k$ as $G_{0}$ //{\tt see Algorithm \ref{algo:findg0}}.

\STATE  $l\leftarrow 0$;

\STATE	\textbf{while}  $\con_{G_{l}}(Q)=$ \textbf{true} \textbf{do}

\STATE  \hspace{0.3cm}  Compute $\dist_{G_{l}}(q, u)$, $\forall q\in Q$ and $\forall u\in G_{l}$;

\STATE  \hspace{0.3cm}  $u^{*} \leftarrow \arg\max_{u\in G_{l}} \dist_{G_l}(u, Q)$;

\STATE  \hspace{0.3cm}  $\dist_{G_l}(G_{l}, Q) \leftarrow \dist_{G_l}(u^*, Q)$;

\STATE  \hspace{0.3cm}  Delete $u^{*}$ and its incident edges from $G_{l}$;

\STATE  \hspace{0.3cm}  Maintain $k$-truss property of $G_l$ //{\tt see Algorithm \ref{algo:simple_trssmt}};

\STATE  \hspace{0.3cm}  $G_{l+1} \leftarrow G_{l}$; $l \leftarrow l+1$;

\STATE  $R \leftarrow \arg\min_{G'\in\{G_0, ..., G_{l-1}\}}{\dist_{G'}(G', Q)}$;

\end{algorithmic}
\end{algorithm}

\subsection{Approximation Analysis}
Algorithm \ref{algo:mindia} can achieve 2-approxiamtion to the optimal solution, that is, the obtained connected $k$-truss community $R$ satisfies $Q\subseteq R$, $\tau(R)= \tau(H^*)$ and $\diam(R) \leq 2\diam(H^*)$, for any optimal solution $H^*$. Since any graph in $\{G_0, ..., G_{l-1}\}$ is a connceted $k$-truss with the largest $k$ containing $Q$ by Algorithm \ref{algo:mindia}, and $R\in \{G_0, ..., G_{l-1}\}$, we have $Q\subseteq R$, and $\tau(R)= \tau(H^*)$. In the following, we will prove that $\diam(R) \leq 2\diam(H^*)$.  We start with a few key results. For graphs $G_1, G_2$, we write $G_1\subseteq G_2$ to mean $V(G_1)\subseteq V(G_2)$ and $E(G_1)\subseteq E(G_2)$.

\begin{fact}\label{lemma.dissubgraph}
Given two graphs $G_1$ and $G_2$ with $G_1\subseteq G_2$, for  $u, v\in V(G_1)$, $\dist_{G_2}(u, v) \leq \dist_{G_1}(u,v)$ holds. Moreover, if $Q\subseteq V(G_1)$, then $\dist_{G_2}(G_1, Q) \leq \dist_{G_1}(G_1, Q)$ also holds.
\end{fact}

\begin{proof}
Trivially follows from the fact that $G_2$ preserves paths between nodes in $G_1$.
\end{proof}

Recall that in Algorithm \ref{algo:mindia}, in each iteration $i$, a node $u^*$ with maximum $\dist(u^*, Q)$ is deleted from $G_i$, but $\dist_{G_i}(G_i, Q)$ is \emph{not} monotone nonincreasing during the process, hence  $\dist_{G_{l-1}}(G_{l-1}, Q)$ is not necessarily the minimum. Note that in Algorithm \ref{algo:mindia}, $G_l$ is not the last feasible graph (i.e., connected $k$-truss containing $Q$), but $G_{l-1}$ is.  The observation is shown in the following lemma.


\begin{lemma}\label{lemma.nondecreasing}
In Algorithm \ref{algo:mindia}, it is possible that for some  $ 0 \leq i < j <l$, we have $G_{j} \subset G_{i}$, and  $\dist_{G_i}(G_{i}, Q) < \dist_{G_j}(G_j, Q)$ hold.
\end{lemma}

\begin{proof}
It is easy to be realized, because for a vertex $v\in G$, $\dist_{G}(v, Q)$ is non-decreasing monotone w.r.t.\ subgraphs of $G$. More precisely, for $v\in G_i \cap G_j$, $\dist_{G_i}(v, Q) \leq \dist_{G_j}(v, Q)$ holds.
\end{proof}

\eat{
\note[Laks]{An example may improve readability. Also, see if you agree with the change I made to the statement of the lemma above.}
\note[Xin]{I will add one more example soon.}
\note[Laks]{BTW, the remark in the proof above is confusing: it says distance function is monotone non-decreasing while above it says explicitly that distance is NOT monotone non-increasing. What is going on?}
\note[Xin] {The writting was not in the accurate description. These two things are not conflict. We say that $\dist_{G_i}(G_i, Q)$ is \emph{not} monotone nonincreasing, the variable is the decreased graph of $G_i$. On the other hand, we say that $\dist_{G}(v, Q)$ fuction is nondecreasing monotone with the decreased graph $G$ for a fixed vertex $v$.}
}

\begin{figure}[t]
\small
\vskip -0.1in
\centering
\includegraphics[width=0.8\linewidth]{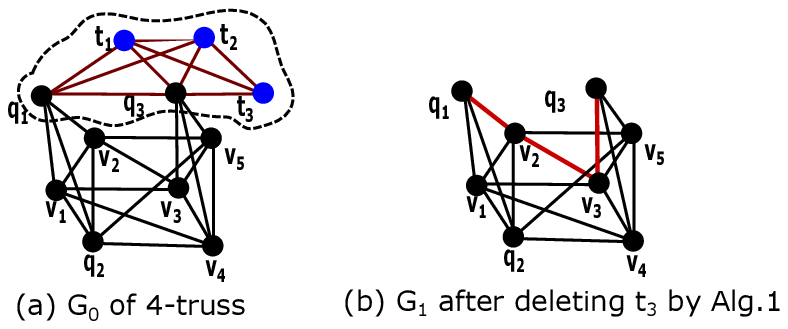}
\vskip -0.1in
\caption{Closest truss community example}\vskip -0.2in
\label{fig.remove}
\end{figure}

\begin{example}
{\em
To illustrate the lemma, suppose the graph in Figure \ref{fig.remove}(a) is $G_0$, a connected 4-truss containing the query nodes $Q=\{q_1\}$ in some initial graph $G$ (not shown) and suppose the maximum trussness of such a subgraph is 4.  
One of furthest nodes from $Q$ in $G_0$ is $t_3$, which has query distance $\dist_{G_0}(t_3, Q) = 2$. After deleting the node $t_3$ from $G_0$, we remove the all incident edges of nodes $t_1$, $t_2$ and $t_3$, since the 4-truss subgraph induced by $\{q_1, q_3, t_1, t_2, t_3\}$ in the dashed region does not exist any more in $G_1$ in Figure \ref{fig.remove}. Thus, we have the largest query distance as $\dist_{G_1}(G_1, Q)=$ $\dist_{G_1}(q_3, Q) =3$, which is larger than $\dist_{G_0}(G_0, Q)=2$.
\qed }
\end{example}

\LL{
We have an important observation that if an intermediate graph $G_{i}$ obtained by Algorithm~\ref{algo:mindia} contains an optimal solution $H^*$, i.e.,  $H^*\subset G_{i}$，and $\dist_{G_i}(G_{i}, Q) > \dist_{G_i}(H^*, Q)$, then  algorithm will not terminate at $G_{i+1}$.

\begin{lemma}\label{lemma.terminate}
In Algorithm \ref{algo:mindia}, for any intermediate graph $G_{i}$, we have  $H^*\subseteq G_{i}$, and $\dist_{G_i}(G_{i}, Q) > \dist_{G_i}(H^*, Q)$, then $G_{i+1}$ is a connected \ktruss containing $Q$ and $H^*\subseteq G_{i+1}$.
\end{lemma}

\begin{proof}
Suppose $H^*\subseteq G_{i}$ and $\dist_{G_i}(G_{i}, Q) $ $> \dist_{G_i}(H^*, Q)$. Then  there exists a node $u\in G_{i} \setminus H^*$ s.t. $\dist_{G_i}(u, Q)$ $ = $ $\dist_{G_i}(G_i, Q)$ $ > \dist_{G_i}(H^*, Q)$. Clearly, $u\notin Q$.
\eat{
Thus, we can obtain that $u\notin Q$ for $Q\subseteq H$.
}
In the next iteration, Algorithm \ref{algo:mindia} will delete $u$ from $G_i$ (Step 7), and perform Step 8. The graph resulting from restoring the \ktruss property is $G_{i+1}$. Since $H^*$ is a connected k-truss containing $Q$,
\eat{and no vertices/edges are deleted in Step 8, thus}
the restoration step (line 8) must find a subgraph $G_{i+1}$ s.t. $H^* \subseteq G_{i+1}$, and $G_{i+1}$ is a connected $k$-truss containing $Q$. Thus, the algorithm will not terminate in iteration $(i+1)$.
\end{proof}
}

We are ready to establish the main result of this section. \RV{Our polynoimal algorithm can find a connected $k$-truss community $R$ having the minimum query distance to $Q$, which is optimal.

\begin{lemma}\label{lemma.optdist}
For any $H$ is a connected $k$-truss with the highest $k$ containing $Q$, $\dist_R(R, Q) \leq \dist_H(H, Q)$.
\end{lemma}

\begin{proof}
The following cases arise for $G_{l-1}$, which is the last feasible graph obtained by Algorithm~\ref{algo:mindia}.

\underline{Case (a)}: $H \subseteq G_{l-1}$.
We have   $\dist_{G_{l-1}}$ $(G_{l-1}, Q) $ $\leq \dist_{G_{l-1}}$ $(H, Q)$; for otherwise, if $\dist_{G_{l-1}}$ $(G_{l-1}, Q) $ $ > \dist_{G_{l-1}} $ $(H, Q)$, we can deduce from Lemma \ref{lemma.terminate} that $G_{l-1}$ is not the last feasible graph obtained by Algorithm~\ref{algo:mindia}, a contradiction. Thus, according to Step 10 in Algorithm \ref{algo:mindia} and $\dist_{G_{l-1}}(G_{l-1}, Q) $ $\leq \dist_{G_{l-1}}(H, Q)$, we have $\dist_R(R, Q)\leq $
$\dist_{G_{l-1}}(G_{l-1}, Q) \leq $ $\dist_{G_{l-1}}(H, Q) \leq $ $ \dist_{H}(H, Q)$.

\underline{Case (b)}: $H\nsubseteq G_{l-1}$.  There exists a vertex $v\in H$ deleted from one of the subgraphs $\{G_0, ..., G_{l-2}\}$. Suppose the first deleted vertex $v^*\in H$ is in graph $G_i$, where $0\leq $ $i\leq l-2$, then $v^*$ must be deleted in Step 7, but not in Step 8. This is because each vertex/edge of $H$ satisfies the condition of $k$-truss, and will not be removed before any vertex is removed from
\LL{$G_i$}. Then, we have $\dist_{G_i}(G_i, Q) $ $= \dist_{G_i}(v^*, Q) $ $= \dist_{G_i}(H, Q)$, and  $\dist_{G_i}(G_i, Q) \geq $ $ \dist_{R}(R, Q)$ by Step 10. As a result,  $\dist_{R}(R, Q) $ $\leq $ $\dist_{G_i}(H, Q)$ $ \leq $ $ \dist_{H}(H, Q)$.
\end{proof}

Based on the preceding lemmas, we have:
\begin{theorem}
Algorithm \ref{algo:mindia} provides a 2-approximation to the \ctcp as $diam(R) \leq 2diam(H^*)$.
\end{theorem}

\begin{proof}
Since $\dist_R(R, Q) \leq \dist_{H^*}(H^*, Q)$ by Lemma \ref{lemma.optdist}, we get $diam(R) \leq $ $ 2\dist_R(R, Q) $ $\leq $ $2\dist_{H^*} (H^*, Q) $ $ \leq $ $2diam(H^*)$ by Lemma \ref{lemma.disdia}. The theorem follows from this.
\end{proof}
}

\subsection{K-truss Identification and Maintenance}
\RV{
In this section, we introduce the detailed implementation of Algorithm~\ref{algo:mindia}. Finding $G_0$, the maximal connected $k$-truss containing $Q$ with the largest trussness $k$, is a basic primitive in our problem. A straightforward method is to apply a truss decomposition algorithm \cite{WangC12}, and delete edges in ascending order of edge support from $G$, until $Q$ becomes disconnected. Then we can obtain the largest trussness $k$ and recover $G_0$ by keeping all $k$-truss edges. However, this method is quite costly. To find $G_0$ efficiently, we design an index structure. The index is constructed by organizing edges according to their trussness.}

\stitle{Index Construction.} We first apply a truss decomposition algorithm such as \cite{WangC12} and compute the trussness of each edge of graph $G$.  \LL{We omit the details of this algorithm due to space limitation.}  

Based on the obtained edge trussness, we construct our truss index as follows. For each vertex $v\in V$, we sort its neighbors $N(v)$ in descending order of the edge trussness $\tau(e(v,u))$, for $u\in N(v)$.  For each distinct trussness value $k\geq 2$, we mark the position of the first vertex $u$ in the sorted adjacency list where $\tau(e(u,v))=k$.  This supports efficient retrieval of $v$'s incident edges with a certain trussness value. The vertex trussness of $v$ is also kept as $\tau(v) = \max\{\tau(v,u)| u\in N(v)\}$, which is the trussness of the first edge in the sorted adjacency list.  Moreover, we build a hashtable to keep all the edges and their trussness values. \LL{This is identical to the simple truss index of \cite{huang2014} and we refer to it as the truss index.}

In the following, we will show that this truss index is sufficient to design an algorithm for \LL{finding the maximal connected $k$-truss containing given query nodes $Q$ in time $O(m')$, where $m' = |E(G_0)|$. This time complexity is essentially \emph{optimal}}.
\eat{Thus, no other complex indexes are not necessary, which will not outperform this time complexity. }
We remark the complexity of this $k$-truss index construction below.

\begin{remark}\label{remark.index}
The construction of this truss index takes $O(\rho \cdot m)$ time and $O(m)$ space, where $\rho$ is the arboricity of graph $G$, i.e., the minimum number of spanning forests needed to cover all edges of $G$. Notice that $\rho \leq \min\{d_{max}, \sqrt{m}\}$ \cite{ChibaN85}. 
\end{remark}

\eat{
\note[Laks]{[22] reports that their truss decomposition algorithm takes $O(m\sqrt{m})$ time. If we use their algorithm, how come the index construction takes less time than that (since arboricity can be smaller)?}
\note[Xin]{The time complexity $O(m^{1.5})$ in [22] is not tight enough for their algorithm. This can be improved to $O(\rho m)$ using other analysis methods without changing algorithm.}
\note[Laks]{Which paper is this other method due to? We should cite it.}
\note[Xin]{I put the paper \cite{ChibaN85}  here.}
}

\begin{algorithm}[t]
\small
\caption{Find$G_0$($G$, $Q$)} \label{algo:findg0}
\textbf{Input:} A graph $G=(V, E)$, a set of query nodes $Q=\{q_1, ..., q_r\}$.\\
\textbf{Output:} A connected $k$-truss $G_0$ containing $Q$ with the largest $k$.\\
\
\begin{algorithmic}[1]

\STATE  $k \leftarrow \min $ $\{\tau(q_1),$ $ ..., \tau(q_r)\}$ //{\tt see Lemma \ref{lemma.largestk}};

\STATE  $V(G_0)\leftarrow \emptyset$; $S_k = Q$;

\STATE	\textbf{while}  $\con_{G_0}(Q)=$ \textbf{false} \textbf{do}

\STATE  \hspace{0.3cm}  \textbf{for} $v \in S_k$ \textbf{do}

\STATE  \hspace{0.3cm}  \hspace{0.3cm}  \textbf{if} $v\in V(G_0)$ \textbf{then}

\STATE  \hspace{0.3cm}  \hspace{0.3cm}  \hspace{0.3cm}  $k_{max} \leftarrow k+1$;

\STATE  \hspace{0.3cm}  \hspace{0.3cm}  \textbf{else}

\STATE  \hspace{0.3cm}  \hspace{0.3cm}  \hspace{0.3cm}  $k_{max} \leftarrow +\infty$; \RV{$V(G_0)\leftarrow V(G_0) \cup \{v\}$;}

\STATE  \hspace{0.3cm}  \hspace{0.3cm}	\textbf{for} $(v, u)\in G$  with $k \leq \tau(v,u)< k_{max}$ \textbf{do}

\STATE  \hspace{0.3cm}  \hspace{0.3cm}	\hspace{0.3cm}	$G_0 \leftarrow G_0 \cup \{(v,u)\}$;

\STATE  \hspace{0.3cm}  \hspace{0.3cm}	\hspace{0.3cm}	\textbf{if} $u\notin S_k$ \textbf{then}
 $S_k \leftarrow S_k \cup \{u\}$;

\STATE  \hspace{0.3cm}  \hspace{0.3cm}	$l \leftarrow \max\{\tau(v, u)| \RV{(v,u)} \notin G_0\}$;

\STATE  \hspace{0.3cm}  \hspace{0.3cm}	$S_l \leftarrow S_l \cup \{v\}$;

\STATE  \hspace{0.3cm}  $k \leftarrow k-1$;

\STATE  Compute the edge support $sup(v,u)$ in $G_0$, for all $(v,u)\in G_0$;

\end{algorithmic}
\end{algorithm}

\stitle{Finding $G_0$.} \RV{Based on the index, we present Algorithm \ref{algo:findg0} for finding $G_0$, the maximal connected $k$-truss containing $Q$ with the largest trussness $k$. We initialize $G_0$ to be the query vertex set $Q$, and iteratively add the edges of $G$ in the decreasing order of trussness, until $G_0$ gets connected.  
} 

The initial trussness level of the edges to be included in $G_0$ is computed as $k = \min\{\tau(q_1),$ $ ..., \tau(q_r)\}$ (line 1). This is motivated by the fact that, by  Lemma~\ref{lemma.largestk}, for any $k'>k$, no connected $k'$-truss can contain $Q$. We use $S_k$ to denote the set of nodes to be visited within level $k$. We start with $S_k = Q$ (line 2).
For a given $k$, we process each node $v\in S_{k}$, and visit its neighbors in a BFS manner. \RV{Then, we insert those of its incident edges $(v, u)$, with $k\leq \tau(v,u)\ \leq k_{max}$ into $G_0$, where $\kmax$ is the maximum possible trussness of unvisited edges.  This is because all these edge should be present in a connected $k$-truss.}
\RV{Meanwhile, if the neighbor $u$ is not in $S_{k}$, we add $u$ into $S_{k}$ (line 11), since unvisited edges incident to $u$ may have trussness no less than $k$.} After checking all edges incident to $v$, we add $v$ to $S_{l}$, where $l = \max\{\tau(v, u)\mid u \in N(v), \tau(v,u) < k\}$ (line 12-13). \RV{Notice that
$l$ is the next highest level for which a connected $l$-truss contains the node $v$, which can avoid scanning the neighbor set of $v$ at each level.}  After traversing all vertices in $S_{k}$, the algorithm checks whether $Q$ is connected in $G_0$. If yes, the algorithm terminates, and $G_0$ is returned; otherwise, we decrease the present level $k$ by 1 (line 14), and repeat the above steps (lines 4-14). After obtaining $G_0$, we compute all edge supports by counting triangles in $G_0$, which is used for the $k$-truss maintenance (line 15). 

The following example illustrates the algorithm.

\begin{example}
\vspace*{-0.2cm}
{\em Consider the graph $G$ in Figure \ref{fig.FindG0} with $Q=\{q_1, $ $ q_2\}$. The trussness of each edge is  displayed, e.g., $\tau(q_1, v_1)=4$. Now, we apply Algorithm \ref{algo:findg0} on $G$ to find $G_0$ containing $Q$. We can verify that $\tau(q_1)=\tau(q_2)=4$ so we start with level $k=4$ and set $S_{4}=\{q_1, q_2\}$. Then, we process the node $q_1\in S_4$, and insert all its incident edges into $G_0$, for the trussness of each edge is 4. Meanwhile, all its neighbors are inserted into $S_4$. We repeat above process for each node in $S_4$. Note that for nodes $t_1$, $t_2$,  $\tau(t_1, t_2)=2$, so we insert $t_1, t_2$ into $S_2$ (lines 11-12 of Algorithm \ref{algo:findg0}). Then, at level $k=4$, we get the 4-truss as the whole graph in Figure \ref{fig.FindG0} minus  the edge $(t_1, t_2)$, for  $\tau(t_1, t_2) =2$. Since the current $G_0$ is not connected, we decrease the truss level $k$ to 3, and find that $S_3 = \emptyset$. Then, we decrease $k$ to $2$, and find that $S_2=\{t_1, t_2\}$. So we expand from the edge incident to $t_1$, and insert the edge $(t_1, t_2)$ into $G_0$, and find that the resulting graph contains $Q$ and is connected. In this example, $G_0$ happens to coincide with $G$.
\qed
}
\end{example}

\begin{figure}[t]
\small
\scriptsize
\vskip -0.1in
\centering
\includegraphics[width=0.45\linewidth]{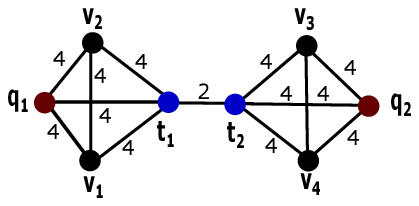}
\vspace{-0.2cm}
\caption{An example graph $G$ of finding $G_0$}
\label{fig.FindG0}
\end{figure}
\vskip -0.2in

\vspace*{-0.1cm}
\begin{remark}
Based on the truss index, for each vertex $v$, in line 9 of Algorithm~\ref{algo:findg0}, each edge $(v,u)$ can be accessed constant time using the sorted adjacent list of $v$, and in line 12, we can compute $l$ in constant time. Algorithm~\ref{algo:findg0} takes time $O(m')$ where $m' = |E(G_0)|$.
\end{remark}

\vspace*{-0.1cm}
\stitle{Computing Query Distance.} For a vertex $v$, to compute the query distance $\dist_{G_i}(v,Q)$, we need to perform $|Q|$ BFS traversals on graph $G_i$. Specifically, for each query node $q\in Q$, with one BFS traversal starting from $q$ in $G_i$, we can obtain the shortest distance $\dist_{G_i}(v, q)$ for each node $v \in G_i$. Then, $\dist_{G_i}(v,Q)$ is the maximum of all shortest distances $\dist_{G_i}(v, q)$, for $q\in Q$.

{
\begin{algorithm}[t]
\small
\caption{\trussmt($G$, $V_d$)} \label{algo:simple_trssmt}
\textbf{Input:} A graph $G=(V, E)$, a set of nodes to be removed as $V_d$.\\
\textbf{Output:} A $k$-truss graph.\\
\
\begin{algorithmic}[1]

\STATE  $S \leftarrow \emptyset$; \RV{//$S$ is the set of removed edges. }

\STATE 	\textbf{for} $v\in V_d$ \textbf{and} $(v, u)\in G$ \textbf{do}

\STATE 	\hspace{0.3cm}  $S \leftarrow S \cup (v, u)$;

\STATE  \textbf{for} $(v, u) \in S$ \textbf{do}

\STATE 	\hspace{0.3cm}  \textbf{for} $w \in N(v) \cap N(u)$ \textbf{do} \RV{\\\hspace{0.3cm} \hspace{0.3cm}  // Update the support of edges $(v, w)$ and $(u,w)$ }

\STATE 	\hspace{0.3cm}  \hspace{0.3cm}	$sup(v, w) \leftarrow sup(v,w)-1$; $sup(u, w) \leftarrow sup(u,w)-1$;

\STATE  \hspace{0.3cm}  \hspace{0.3cm}  \textbf{if} $sup(v, w) < k-2$ \textbf{and} $(v,w)\notin S$ \textbf{then} $S \leftarrow S \cup (v, w)$;

\STATE  \hspace{0.3cm}  \hspace{0.3cm}  \textbf{if} $sup(u, w) < k-2$ \textbf{and} $(u,w)\notin S$ \textbf{then} $S \leftarrow S \cup (u, w)$;

\STATE  \hspace{0.3cm}  Remove $(v, u)$ from $G$;

\STATE  Remove isolated vertices from $G$;

\end{algorithmic}
\end{algorithm}
}

{

\stitle{K-truss Maintenance.} Algorithm \ref{algo:simple_trssmt} describes the procedure for maintaining $G$ as a $k$-truss after the deletion of nodes $V_d$ from $G$. In Algorithm~\ref{algo:mindia}, $V_d=\{u^*\}$ (see line 8).\footnote{In Section~\ref{sec.fast}, we will discuss deleting a set of nodes $V_d$ in batch.} \RV{Generally speaking, after removing nodes $V_d$ and their incident edges from $G$, $G$ may not be a $k$-truss any more, or $Q$ are disconnected. Thus, Algorithm \ref{algo:simple_trssmt} iteratively deletes edges having less than $(k-2)$ triangles and nodes disconnected with $Q$ from $G$, until $G$ becomes a connected $k$-truss containing $Q$. }

 Algorithm \ref{algo:simple_trssmt} firstly pushes all edges incident to nodes $V_d$ into set $S$ (lines 1-3). Then, for each edge $(u,v)\in S$, the algorithm checks every triangle $\triangle_{uvw}$ where $w\in N(u)\cap N(v)$, and decreases the support of edges $(u,w)$ and $(v,w)$ by 1; For any edge $e \notin S$, with resulting support $\sup(e)< k-2$, $e$ is added to $S$. After traversing all triangles containing $(u,v)$, the edge $(u,v)$ is deleted from $G$. This process continues until $S$ becomes empty (lines 4-9), and then the algorithm removes all isolated vertices form $G$ (line 10).
}

\subsection{Complexity analysis}
\LL{In the implementation of Algorithm \ref{algo:mindia}, we do not need to keep all immediate graphs, but just record the removal of vertices/edges at each iteration. Let $G_0$ be the maximal connected $k$-truss found in line 1 of Algorithm~\ref{algo:mindia}. Let $n'=|V(G_0)|$ and $m'=|E(G_0)|$, and let $d'_{max}$ be the maximum degree of a vertex in $G_0$.
}

At each iteration $i$ of Algorithm \ref{algo:mindia}, we delete at least one node and its incident edges from $G_i$. Clearly, the number of removed edges is no less than $k-1$, thus the total number of iterations is $t \leq \min\{n'-k, m'/(k-1)\}$, i.e., $t$ is $O(\min\{n', m'/k\})$.  We have:

\begin{theorem}\label{theorem.framework}
Algorithm \ref{algo:mindia} takes $O((|Q|t+\rho)m')$ time and  $O(m')$ space, where $t\in O(\min\{n', m'/k\})$, and $\rho$ is the arboricity of graph $G_0$. Furthermore, we have $\rho \leq \min\{d'_{max}, \sqrt{m'}\}$.
\end{theorem}

\begin{proof}
First, finding the $k$-truss $G_0$, listing all triangles of $G_0$ and creating a series of $k$-truss graphs $\{G_0, ..., G_{l-1}\}$  takes $O(\rho\cdot m')$ time in all, where $\rho$ is the arboricity of graph $G_0$.

Second, in each iteration, the algorithm needs to compute the shortest distances by a BFS traversal  strating from each query node $q\in Q$, which takes $O(|Q| m')$ time. Since the algorithm runs in $t$ iterations, the total time cost is $O(t |Q| m')$. Thus, the overall time complexity of Algorithm \ref{algo:mindia} is $O((|Q|t+\rho)m')$.

Next, we analyze the space complexity. For graphs $\{G_0, ..., G_l\}$,  we only record the sequence of removed edges from $G_0$ for attaching a corresponding label to a graph $G_i$ at each iteration $i$, which takes $O(m')$ space in all.  For each vertex $v\in G_i$, we only keep $\dist(v, Q)$ instead of all query distances $\dist(v, q)$ for $q\in Q$, which takes $O(n')$ space. \LL{Hence, the space complexity of Algorithm \ref{algo:mindia} is $O(m'+n')$, which is  $O(m')$, as $G_0$ is connected}.
\end{proof}

%% file: fast.tex
In this section, we focus on improving the efficiency of CTC search in two ways. 
First, we develop a new greedy strategy to speed up the pruning process process in Section \ref{sec.lazy}, by deleting at least $k$ nodes in batch, to achieve quick termination while sacrificing some approximation ratio. Second, we also propose a heuristic strategy to quickly find the \ctc in the local neighborhood of query nodes.

\subsection{Bulk Deletion Optimization} \label{sec.lazy}
In this subsection,  we propose a new algorithm called \lazy following the framework of Algorithm \ref{algo:mindia}, which is based on deletion of a set of nodes in batch when maintaining a \ktruss. The algorithm is described in detail in Algorithm \ref{algo:lazy}, which can terminate quicker than Algorithm \ref{algo:mindia}. It is based on the following two observations. 

First, in Algorithm \ref{algo:mindia}, if a graph $G_i$ has query distance $\dist_{G_i}(G_i, Q)$ $ = d$, only one vertex $u^*$ with $\dist_{G_i}(u^*, Q)=d$ is removed from $G_i$. Instead, we can delete \emph{all} nodes $u$ with $\dist_{G_i}(u, Q)=d$, from $G_i$, in one shot. The reason is that $\dist_{G_i}(u,Q)$ is monotone non-decreasing with decreasing graphs, i.e., $dist_{G_j}(u,Q)$ $ \ge $ $ \dist_{G_i}(u,Q)$ $ = d$, for $j>i$. 
Thus, removing a set of vertices $L = \{u^*| \dist_{G_i}(u^*, Q)\geq d, u^*\in G_{i}\}$ in each iteration $i$ will improve the efficiency. This improvement indeed works in real applications. However, in theory, it is possible that $|L|=1$ in every iteration. 

\TobeDelete{Delete the lemma and explain it with a sentence.}
\eat{
Our second observation, is that a vertex $u^*$ with $\dist_{G_i}(u^*, Q)$ $=d$ has at least $k-1$ neighbors $v$ with $\dist_{G_i}(v, Q)=d-1$. If we remove $L = \{u | $ $\dist_{G_i}(u, Q)$ $\geq d-1,$ $ u\in G_{i}\}$ at each iteration, then at least $k$ nodes are removed each time. Thus, the resulting number of iterations is $O(n'/k)$, where $n'=|V(G_0)|$.}

\RV{
\LL{Our second observation, shown in the next lemma, is that a vertex $u^*$ with $\dist_{G_i}(u^*, Q)$ $=d$ has at least $k-1$ neighbors $v$ with $\dist_{G_i}(v, Q)=d-1$. If we remove $L = \{u | $ $\dist_{G_i}(u, Q)$ $\geq d-1,$ $ u\in G_{i}\}$ at each iteration, then the resulting number of iterations is $O(n'/k)$, where $n'=|V(G_0)|$. }  


\begin{lemma}
Algorithm \ref{algo:lazy} terminates in  $O(n'/k)$ iterations.
\end{lemma}

\begin{proof}
In Algorithm \ref{algo:lazy}, at each iteration $i$, the graph $G_i$ has at least one node $u^*$ with $\dist_{G_l}(u^*, Q) = d$, which belongs to $L$, and will be deleted in this iteration (lines 4-10). Since $G_i$ is a connected k-truss and $u^*\in G_i$, $u^*$ has at least $k-1$ neighbors, i.e., $|N_{G_i}(u^*)| \geq k-1$. Moreover, $\forall v\in N_{G_i}(u^*)$, we have $\dist_{G_i}(v, Q) \geq d-1$: otherwise, if $\exists v\in N_{G_i}(u^*)$ with $\dist_{G_i}(v, Q) < d-1$, we can obtain $\dist_{G_i}(u^*, Q) < d$, a contradiction. As a result, we have $u^*\in L$ and $N(u) \subset L$, and $|L|\geq k$. Thus, at least $k$ nodes are deleted at each iteration, and the algorithm terminates in $O(n'/k')$ iterations. 
\end{proof}

}

Thus, the number of iterations is improved from $O(\min\{n', m'/k\})$ to 
$O(n'/k)$ (see Theorem \ref{theorem.framework}). We just proved: 

\begin{theorem} \label{theorem.fast}
Algorithm \ref{algo:lazy} takes $O((|Q|t'+\rho')m')$ time using $O(m')$ space, where $t' \in O(n'/k)$, and $\rho' \leq \min\{d'_{max},$ $ \sqrt{m'}\}$.
\end{theorem}


The approximation quality of Algorithm \ref{algo:lazy} is characterized below. 

\begin{theorem}
Algorithm \ref{algo:lazy} is a $(2+\varepsilon)$-approximation solution of \ctcp, where $\varepsilon = 2/diam(H^*)$. 
\end{theorem}

\LL{
\begin{proof}
To prove this theorem, we only need to ensure $\dist_R(R, Q)$ $ \leq $ $\dist_{H^*}(H^*, Q)+1$. Because $diam(R) $ $\leq $ $ 2\dist_R(R, Q) $ $\leq $ $2\dist_{H^*} (H^*, Q) +2$ $ \leq $ $2(diam(H^*)+1)$ by Lemma \ref{lemma.disdia}, then approximation ratio is $2+\varepsilon$, where $\varepsilon = 2/diam(H^*)$. The detailed proof is similar with Lemma \ref{lemma.optdist}, which is omitted here, due to space limitation.
\end{proof}
}


\begin{algorithm}[t]
\small
\caption{\lazy($G$, $Q$)} \label{algo:lazy}
\textbf{Input:} A graph $G=(V, E)$, a set of query nodes $Q=\{q_1, ..., q_r\}$.\\
\textbf{Output:} A connected $k$-truss $R$ with a small diameter.\\
\
\begin{algorithmic}[1]

\STATE  Find$G0$ (G, Q) //{\tt see Algorithm \ref{algo:findg0}};

\STATE  $d \leftarrow +\infty$;  $l\leftarrow 0$;

\STATE	\textbf{while}  $\con_{G_{l}}(Q)=$ \textbf{true} \textbf{do}

\STATE  \hspace{0.3cm}  Compute $\dist_{G_{l}}(q, u)$, $\forall q\in Q$ and $\forall u\in G_{l}$;

\STATE  \hspace{0.3cm}  $\dist_{G_l}(G_{l}, Q) \leftarrow \max_{u^*\in G_l}\dist_{G_l}(u^*, Q)$;

\STATE  \hspace{0.3cm}  \textbf{if} $\dist_{G_l}(G_{l}, Q) < d$ \textbf{then}

\STATE  \hspace{0.3cm}  \hspace{0.3cm}  $d \leftarrow \dist_{G_l}(G_{l}, Q)$;

\STATE  \hspace{0.3cm}  $L = \{u^*| \dist_{G_l}(u^*, Q)\geq d-1, u^*\in G_{l}\}$;

\STATE  \hspace{0.3cm}  Maintain $k$-truss property of $G_l$ //{\tt see Algorithm \ref{algo:simple_trssmt}};

\STATE  \hspace{0.3cm}  $G_{l+1} \leftarrow G_{l}$; $l \leftarrow l+1$;

\STATE  $R \leftarrow \arg\min_{G'\in\{G_0, ..., G_{l-1}\}}{\dist_{G'}(G', Q)}$;

\end{algorithmic}
\end{algorithm}

\begin{example}
{\em 
Continuing with the previous example, we apply Algorithm \ref{algo:lazy} on Figure \ref{fig.community}(a) to find the closest truss community for $Q=\{q_1, q_2, q_3\}$. In $G_0$, we compute $d = $ $\max_{u\in G_0}\dist_{G_0}(u, Q)$ $= 4$, and $L=\{q_1, q_3, p_1, p_2, p_3\}$, as each node $u\in L$ has query distance  $\dist_{G_0}(u, Q) =3 $ $\geq d-1$. After removing $L$ from $G_0$, the remaining graph does not contain $Q$, and the algorithm terminates. Thus, Algorithm \ref{algo:lazy} reports the entire 4-truss $G_0$ as the answer, which has  diameter 4, compared to the answer of Figure~\ref{fig.community}(b) reported by Algorithm~\ref{algo:mindia}, which has diameter 3. \qed 
} 
\end{example}

%% file: local.tex
In this subsection, we develop a heuristic strategy to quickly find the \ctc by local exploration. The key idea is as follows. We first form a Steiner tree to connect all query nodes, and then expand it to a graph $G_{0}'$ by involving the local neighborhood of the query nodes. From this new graph $G_{0}'$, we find a connnected $k$-truss with the highest $k$ containing $Q$, and then  iteratively remove the furthest nodes from this $k$-truss using the \lazy algorithm discussed earlier.

\begin{algorithm}[t]
\small
\caption{\local($G$, $Q$)} \label{algo:local}
\textbf{Input:} A graph $G=(V, E)$, a set of query nodes $Q=\{q_1, ..., q_r\}$, a node size threshold $\eta$.\\
\textbf{Output:} A connected $k$-truss $R$ with a small diameter.\\
\
\begin{algorithmic}[1]

\STATE  Compute a Steiner Tree $T$ containing $Q$ using truss distance fuctions;

\STATE  $k_{t} \leftarrow \min_{e\in T} \tau(e)$;

\STATE  Expand $T$ to a graph $G_{t} = \{e\in G |$ $\tau(e)\geq k_{t}\}$, s.t. $T\subseteq G_t$ and $|G_t|\leq \eta$;

\STATE 	Extract the maximal connected $k$-truss $H_{t}$ containing $Q$ from $G_{t}$, where $k\leq k_t$ is the maximum possible trussness;

\STATE 	Apply \lazy algorithm on $H_{t}$ to identify closest community.

\end{algorithmic}
\end{algorithm}

\stitle{Connect query nodes with a Steiner tree.} As explained above, the Steiner tree found is used as a seed for expanding into a $k$-truss. It is well-known that finding a minimal weight Steiner tree is NP-hard but it admits a 2-approximation \cite{kou1981fast,mehlhorn1988faster}. However, a naive application of these algorithms may produce a result with a small trussness. To see this, consider the graph $G$ and the query $Q=\{q_1, q_2, q_3\}$ in Figure \ref{fig.community}(a). Suppose all edges are uniformly weighted. Then it is obvious that the tree $T_1 = \{(q_2, q_1), $ $(q_1, t), $ $(t, q_3)\}$ with total weight 3 is an optimal (i.e., minimum weight) Steiner tree for $Q$. However, the smallest trussness of the edges in $T_1$ is 2, which suggests growing $T_1$ into a larger graph will yield a low trussness. By contrast, the Steiner tree $T_2＝\{(q_1, q_2), $ $(q_2, v_4), $ $(v_4, q_3)\}$ has the total weight 3 and all its edges have the trussness at least 4, indicating it could be expanded into a more dense graph. To help discriminate between such Steiner trees, we define path weights as follows. Recall the definition of $\taubar(S)$ from Section~\ref{sec.problem}.

\eat{
In this algorithm, the choice of the Steiner tree containing $Q$, which is the most improtant step of all. Since we select it as a seed for expanding to large graphs.   In weighted graphs, to find a connected spanning tree contianing all query nodes with the lowest weight, which is a well-known Steiner Tree problem \cite{kou1981fast}. There has already existed sevearl fast solutions \cite{kou1981fast,mehlhorn1988faster} with no greater than 2-approxiamtion guarantee. However, in our \ctc model, we not only requires the distance beween query nodes should be small, but also  need these query nodes to be densely connected in $k$-truss with large $k$. Thus, the algorithms can not be straight-forward applied in our problem, which may causes the Steiner tree conncted by the edges of small trussness. For example, consider graph $G$ and the query $Q=\{q_1, q_2, q_3\}$ in Figure \ref{fig.community}(a), it is obviously that the tree $T_1 = \{(q_2,$ $ q_1), $ $(q_1, t),$ $ (t, q_3)\}$ with the total weight of 3 is a optimal Steiner tree for $Q$, but the smallest trussness of these tree edges is 2. Thus, we prefer the Steiner tree $T_2＝$ $\{(q_1, q_2),$ $(q_2, v_4), $ $(v_4, q_3)\}$ more, since the total weight are also 3 and all tree edges has the trussness at least 4, indicating a more dense graph it can be expaned to. Hence, we modified the distance function as follow, by giving penalty for the edge of small trussness in the path.
}

\begin{definition}
[Truss Distance] Given a path $P$ between nodes $u, v$ in $G$, we define the truss distance of $u$ and $v$ as $\hat{\dist}_{P}(u, v)$ $ = $ $\dist_{P}(u,v)+ $ $\gamma (\taubar(\emptyset) -\min_{e\in P} \tau(e))$, where $\dist_{P}(u, v)$ is the path length of $P$, and $\gamma > 0$. For a tree $T$, by $\hat{\dist}_T(u,v)$ we mean $\hat{\dist}_P(u,v)$ where $P$ is the path connecting $u$ and $v$ in $T$.
\end{definition}\label{def.trussdist}

The difference $\taubar(\emptyset) -\min_{e\in P} \tau(e)$ measures how much the minimum edge trussness of path $P$ falls short of the maximum edge trussness of the graph $G$ and $\gamma$ controls the extent to which small edge trussness is penalized. The larger $\gamma$ is, the more important edge trussness is in distance calculations. \RV{Note that, for a special path $P$ of a single edge $(u,v)$, the minimum edge truss in $P$ is $\tau(u,v)$. On the other hand, for a path $P$ of length more than 1, the penalty only depends on the minimum edge trussness of path $P$, but not accounts for every edge in $P$. In order to leverage the well-known approximation algorithm of Steiner tree algorithm \cite{mehlhorn1988faster}, we define the truss distance for a path.    
Recall the procedure of Steiner tree algorithm \cite{mehlhorn1988faster}, given a graph $G$ and query nodes $Q$, it firstly constructs a complete distance graph $G'$ of query nodes where the distance equals to its shortest path length in $G$, and finds a minimum spanning tree $T$ of $G'$, then constructs another graph $H$ by replacing each $T's$ tree edge by its corresponding shortest path in $G$, and finally finds a minimum spanning tree of $H$ and deleting leaf edges. We apply the truss distance function on the path weight for shortest path and minimum spanning tree construction here.}
\eat{
The parameter $\gamma$ contorls the penalty weight of small edge trussness, and $\taubar(\emptyset)$ is the maximum trussenss in graph $G$, which keeps the penalty non-negative. The larger $\gamma$ is, the more important of edge trussness in the truss distance calculations.
}
For instance, in the above example, $\taubar(\emptyset) = 4$ and for $\gamma=3$, the truss distance of $(q_2, q_3)$ in $T_1$ is $\hat{\dist}_{T_1}(q_2, q_3) = \dist_{T_1}(q_2, q_3) +3\cdot (4-2) = 3+6 = 8$, since the minimum edge trussness of $T_1$ is $\tau(q_1,t)=2$. On the other hand, $\hat{\dist}_{T_2}(q_1, q_3) = \dist_{T_2}(q_1, q_3) +3\cdot (4-4) = 3+0 = 3$. Obviously, the Steiner tree  $T_2$ has a smaller truss distances than $T_1$.
It can be verified that its overall weight is smaller than that of $T_1$.

\stitle{Find $G_0$ by expanding Steiner tree to graph.} After obtaining the Steiner tree $T$ for the query nodes, we locally expand the tree to a small graph $G_t$ as follows. We firstly obtain the minimum trussness of edges in $T$ as $k_{t} = \min_{e\in T} \tau(e)$. Then, we start from the nodes in $T$, and expand the tree to a graph in a BFS manner via edges of trussness no less than $k_t$, and iteratively insert these edges into $G_t$ until the node size exceeds a threshold $\eta$, i.e., $|V(G_t)| \leq \eta$, were $\eta$ is empirically tuned. Since $G_t$ is a local expansion of $T$, the trussness of $G_t$ will be at most $k_{t}$, i.e., $\tau(G_t) \leq k_t$. For ensuring the  dense cohesive structure of identified communities, we apply a truss decompostion algorithm on $G_{t}$. Then, we extract the maximal connected $k$-truss subgraph $H_t$ containing $Q$ by removing all edges of trussness less than \RV{$k$ from $G_t$}, where $k \leq k_t$ is the maximum possible trussness.

\stitle{Reduce the diameter of $G_0$.} We take the graph $H_t$ with the maximum trussness $k$ as input, and apply a variant of \lazy algorithm on $H_{t}$ for returning the identified community.  \RV{We implement a variant of \lazy algorithm, which is different from original \lazy w.r.t. the removed vertex set $L = \{u^*| \dist_{G_l}(u^*, Q)\geq d-1, u^*\in G_{l}\}$. We readjust the furthest nodes to be removed, as $L'=\{u^*| \dist_{G_l}(u^*, Q)\geq d, u^*\in G_{l}\}$. \LL{This adjustment makes the algorithm not as efficent as \lazy in asymptotic running time complexity, but we still find it efficient in practice. On the other hand, in practice, this strategy can achieve a smaller graph diameter than \lazy. This new strategy provides a 2-approximation of the optimal.}}
\eat{
\note[Laks]{Why is this more efficient in practice? Previously we said that removing nodes with distance $\geq d-1$ makes it terminate faster.}
\note[Xin]{In \lazy algorithm, we remove $\geq d-1$ nodes, since it has the theory gurantee. Here, toe remove $\geq d$ nodes, will be slow than \lazy, but achieve a good approximation and obtain a small diameter.}
}
\LL{
Moreover, in our implementation, in each iteration, we carefully remove only a subset of nodes in $L'$, which have \RV{the largest total of distances from} all query nodes. As a result, more nodes with the largest query distance are removed from the community in the end. The reason is as follows. Suppose the largest query distance we found as $d$,  in the real world, 
the number of nodes having query distance $d$ may be large, due to the small-world property. }

%% file: exp.tex
We conduct experimental studies using 6 real-world networks
%
%
%
available from the Stanford Network Analysis
Project\footnote{\small{\url{snap.stanford.edu}}}, where all networks
are treated as undirected.  The network statistics are shown in
Table~\ref{tab:dataset}. All networks except for Facebook contain
5,000 top-quality ground-truth communities.

\begin{table}[t]
\begin{center}\vspace*{-0.6cm}
\scriptsize
\caption[]{Network statistics (K $=10^3$ and M $=10^6$)}\label{tab:dataset}
\begin{tabular}{|l|r|r|r|r|}
\hline
{\bf Network} & $|V_G|$ & $|E_G|$ & $d_{max}$ & $\taubar(\emptyset)$ \\
\hline \hline
Facebook	 & 4\textbf{K}	 & 88\textbf{K} & 1,045  & 97 \\ \hline
Amazon	& 335\textbf{K} &	926\textbf{K} &	549 &	7\\ \hline
DBLP 	 & 317\textbf{K}	 & 1\textbf{M} & 342  &114 \\ \hline
Youtube & 1.1\textbf{M}	 & 3 \textbf{M}		 & 28,754 & 19\\ \hline
LiveJournal & 4\textbf{M} & 35\textbf{M} & 14,815 & 352 \\ \hline
Orkut	 & 3.1\textbf{M}	 & 117\textbf{M} & 33,313 & 78 \\ \hline
\end{tabular}\vspace*{-0.6cm}
\end{center}
\vspace*{-0.3cm}
\end{table}


\eat{
\begin{table}[t]
\begin{center}\vspace*{-0.6cm}
\scriptsize

\caption[]{\textbf{Parameters of Query Nodes}}\label{tab:parameters}
\begin{tabular}{|l|c|c|}

\hline
{\bf Parameter} & {\bf Range} & {\bf Default} \\
\hline \hline
query size $|Q|$	 &  1, 2, 4, 8, 16  & 3 \\ \hline
degree rank $Q_{d}$ &  20\%, 40\%, 60\%, 80\%, 100\%  & 80\% \\ \hline
inter-distance $l$ & 1, 2, 3, 4, 5 & 2\\ \hline

\end{tabular}\vspace*{-0.5cm}
\end{center}
\end{table}

}



To evaluate the efficiency and effectiveness of improved strategies,
we test and compare three algorithms proposed in this paper, namely,
%
%
%
\textbf{\CTC}, \textbf{\LAZY}, and \textbf{\LCTC}.
%
%
Here, \textbf{\CTC} is the basic greedy approach \framework in
Algorithm \ref{algo:mindia}, which removes single furthermost node at
each iteration.
%
%
%
%
\textbf{\LAZY} is the \lazy approach in Algorithm \ref{algo:lazy},
which removes multiple furthermost nodes at each iteration.
%
%
\textbf{\LCTC} is the local exploration approach in
Algorithm \ref{algo:local}.  For \LCTC, we set the parameters $\eta =
1,000$ and $\gamma = 3$, where \JY{
$\eta=1,000$ is selected to achieve
stable quality and efficiency by testing $\eta$ in $[500, 2,000]$, and
$\gamma = 3$ is selected to balance the requirements of trussness and
diameter for communities searched.
}


%
%
%
%
%

We randomly generate sets of query nodes to test.
\eat{
Table \ref{tab:parameters} shows the three parameters, query size
$|Q|$, degree rank $Q_d$, and inter-distance $l$, used for generating
query nodes with their ranges and default values.
}
\RV{Three parameters, query size
$|Q|$, degree rank $Q_d$, and inter-distance $l$, are used for generating
query nodes with varied values.
}
Here, $|Q|$ is the
number of query nodes, which is set to 3 by default. $Q_{d}$ is the
degree rank of query nodes. We sort all vertices in descending
order of their degrees in a network. A node is said to be with degree
rank of X\%, if it has top highest X\% degree in the network. The default
value of $Q_d$ is 80\%, which means that a query node has degree
higher than the degree of 20\% nodes in the whole network. The
inter-distance $l$ is the inter-distance between all query nodes. The
default $l=2$ indicates that all query nodes are within distance of 2
to each other in the network.

For the efficiency, we report runtime in seconds. We treat the runtime
of a query as infinite if its runtime exceeds 1 hour.

\JY{
For the effectiveness of eliminating ``free riders'', we compare our
methods with \TRUS (Algorithm \ref{algo:findg0}), which finds the
connected $k$-truss graph containing query nodes with the largest $k$
only.
%
%

Let $G_R$ be the \ctc found by \LCTC and $G_0$ be computed by \TRUS.
We report two things. One is the percentage of nodes that are kept in
the resulting community by $\frac{|V(G_R)|}{|V(G_0)|}$.  The less
percentage the more ``free riders'' being removed. The other is the
edge density $2|E(g)|/|V(g)|(|V(g)|-1)$, where $g$ is either $G_R$ or
$G_0$.}

\RV{In addition, to evaluate the quality of \ctc found, we implemented two
state-of-the-art community search methods: the minimum
degree-based community search (\MDC) \cite{sozio2010}, which globally
finds the dense subgraph containing all query nodes with the highest
minimum degree under the distance and size constraints, and the query
biased densest community search (\QDC) \cite{wu2015robust}, which
shifts the detected community to the neighborhood of the query by
integrating the edge density and nodes proximity to the query
nodes.}
\JY{
Here, \MDC and \QDC are implemented using the same data structures,
such as graph, steiner tree and hashtable as we do for \LCTC. To
compare \LCTC with \MDC and \QDC, we test the datasets with
ground-truth, and show F1-score to measure the alignment between a
discovered community $C$ and a ground-truth community $\hat{C}$. Here,
$F1$ is defined as $F1(C, \hat{C})$ $=\frac{2\cdot
prec(C, \hat{C}) \cdot recall(C, \hat{C})}{prec(C, \hat{C}) +
recall(C, \hat{C})}$
where $prec(C, \hat{C}) =\frac{|C\cap \hat{C}|}{|C|}$ is the precision
and $recall(C, \hat{C}) =\frac{|C\cap \hat{C}|}{|\hat{C}|}$ is the
recall.
}



%
%
%

All algorithms are implemented in C++, and all the experiments are
conducted on a Linux Server with Intel Xeon CUP X5570 (2.93 GHz) and
50GB main memory.

\begin{figure}[t]
 \vskip -0.4cm
\centering \mbox{
\subfigure[Query Time]{\includegraphics[width=0.37\linewidth,height=2.5cm]{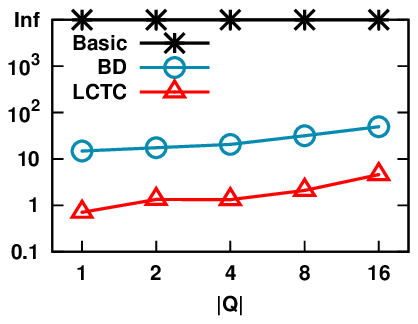}}\hskip -0.15in
\subfigure[The percentage]{\includegraphics[width=0.37\linewidth,height=2.5cm]{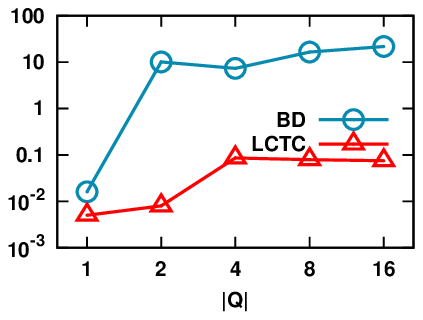}}
\hskip -0.15in
\subfigure[Density]{\includegraphics[width=0.37\linewidth,height=2.5cm]{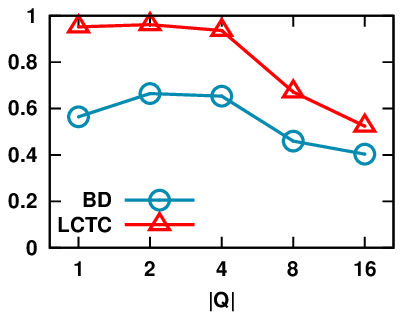}}
}
\vskip -0.2in
\caption{DBLP: varying query size $|Q|$} \label{fig.size_time}
\vspace*{-0.3cm}
\end{figure}

\begin{figure}[t]
\vskip -0.2cm
\centering \mbox{
\subfigure[Query Time]{\includegraphics[width=0.37\linewidth,height=2.5cm]{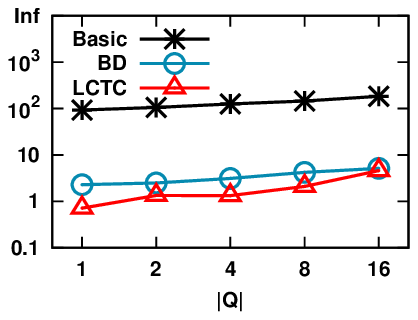}}\hskip -0.15in
\subfigure[The percentage]{\includegraphics[width=0.37\linewidth,height=2.5cm]{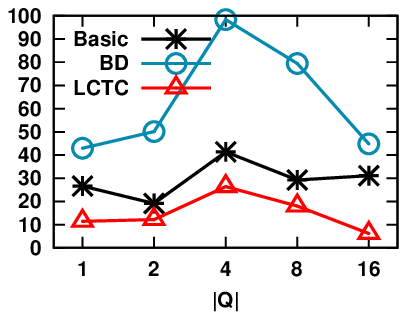}}
\hskip -0.15in
\subfigure[Density]{\includegraphics[width=0.37\linewidth,height=2.5cm]{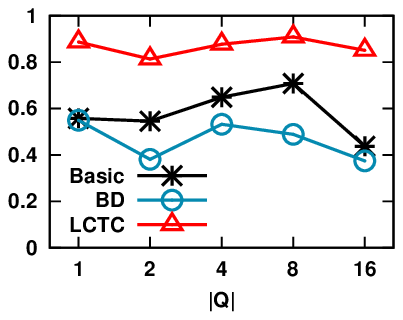}}
}
\vskip -0.2in
\caption{Facebook: varying query size $|Q|$} \label{fig.size_V}
\vspace*{-0.3cm}
\end{figure}


\noindent
{\bf Exp-\RV{1}\eat{2} Different Queries}: We test our approaches
using different queries on DBLP and Facebook
(Table~\ref{tab:dataset}).

First, we vary the query size $|Q|$. We test 5 different $|Q|$ in \{1,
2, 4, 8, 16\}. \RV{For each $|Q|$, we randomly select 100 sets of
$|Q|$ query nodes, and we report the average runtime, the average
percentage of avoiding \FRE and the average edge
density.}\footnote{Notice that \CTC, \LAZY, \LCTC are not optimal
algorithms.} The results for DBLP and Facebook are shown in
Figure \ref{fig.size_time} and Figure \ref{fig.size_V}, respectively.
\JY{\LCTC
outperforms the best in terms of efficiency, the percentage of
avoiding \FRE, and edge density in all cases.}
\CTC cannot find communities in DBLP in 1 hour limit.
%
%
\JY{\LAZY achieves better efficiency in Facebook than DBLP.
This is because Facebook contains only 4K vertices and
the global method \LAZY is effective on such a small network.
%
However, \LAZY performs worse than \CTC
for the percentage of avoiding \FRE and density for Facebook.}
%

\comment{
Figure \ref{fig.size_V}
shows the averaged the percentage of $\frac{|V(R)|}{|V(G_0)|}$ of each
method by varying the parameter $|Q|$. As we can see that, \LCTC keeps
the least nodes in the community among all three methods. \CTC is
without output communities on DBLP, which is not shown in
Figure \ref{fig.size_V} (b). \LAZY achieves the worst performance on
both network. All methods increase the percentage of remaining nodes
w.r.t. the size of connected $k$-truss with the largest $k$, by the
increased parameter query size. Only less than 20\% nodes and 0.1\%
nodes respectively for Facebook and DBLP, the method \LCTC obtained,
which shows the powerful of avoiding \FRE by our truss community
model.
}

\begin{figure}[t]
\vskip -0.4cm
\centering \mbox{
\subfigure[Query Time]{\includegraphics[width=0.37\linewidth,height=2.5cm]{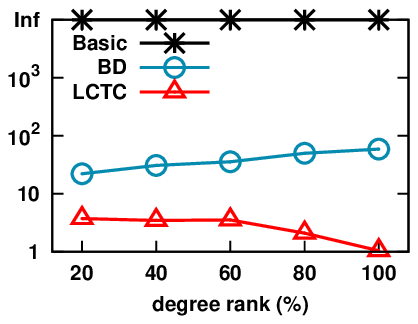}}\hskip -0.15in
\subfigure[The
percentage]{\includegraphics[width=0.37\linewidth,height=2.5cm]{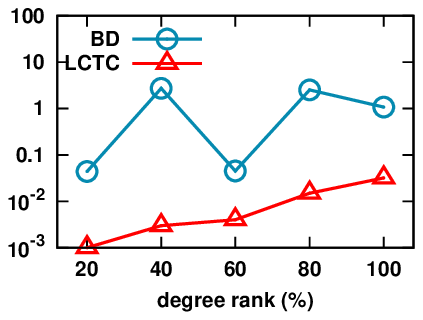}}
\hskip -0.15in
\subfigure[Density]{\includegraphics[width=0.37\linewidth,height=2.5cm]{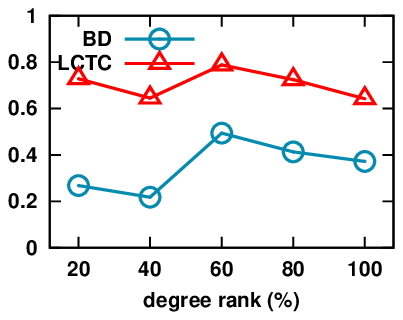}}
}
\vskip -0.2in
\caption{DBLP: varying query vertices} \label{fig.deg_time}
\vspace*{-0.3cm}
\end{figure}

\begin{figure}[t]
\vskip -0.2cm
\centering \mbox{
\subfigure[Query Time]{\includegraphics[width=0.37\linewidth,height=2.5cm]{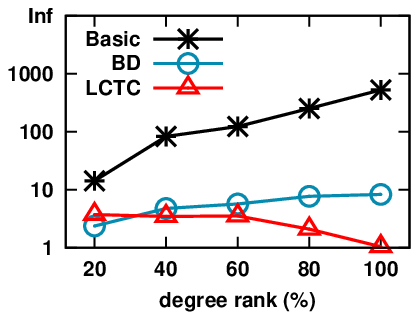}}\hskip -0.15in
\subfigure[The percentage]{\includegraphics[width=0.37\linewidth,height=2.5cm]{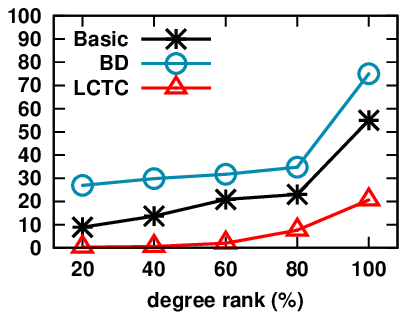}}
\hskip -0.15in
\subfigure[Density]{\includegraphics[width=0.37\linewidth,height=2.5cm]{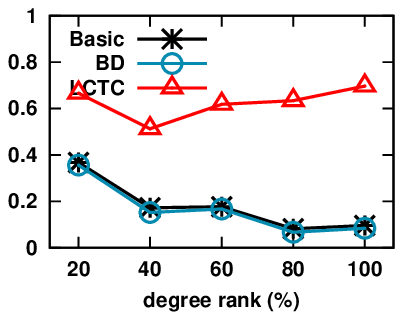}}
} \vskip -0.2in
\caption{Facebook: varying query vertices} \label{fig.deg_V}
\vspace*{-0.3cm}
\end{figure}

Second, we vary the degree of query nodes.  For a graph to be tested,
we sort the vertices in descending order of their degrees, and
partition them into 5 equal-sized buckets. \RV{ For each bucket, we
randomly select 100 different query sets of size 3, and we report the
average runtime, the average percentage of avoiding \FRE and the average density.} The
results for DBLP and Facebook are shown in Figure \ref{fig.deg_time}
and Figure \ref{fig.deg_V}, respectively. \RV{  In terms of runtime, the percentage of avoiding \FRE and density, the performance are similar to the results
by varying the query sizes.}
\LCTC outperforms the others.

\comment{
The average query time for each degree group is
reported in Figure \ref{fig.deg_time}. It is observed that global
methods \CTC and \LAZY cost more query processing time when the degree
of query nodes decreases, and \LCTC is inverses. For the methods \CTC
and \LAZY, the connected $k$-truss containing query nodes of low
degree may have small trussness $k$ even for the largest one $k$, and
the size of corresponding graph $G_0$ would be larger, which leads to
a more expensive computation for community search. On the contrary,
the method \LCTC considers a smaller local graph for query nodes of
lower degree, and the query processing becomes more efficient. Once
again, \LCTC outperforms than other methods on two networks, and \CTC
is still the slowest of all. \LL{ The averaged the percentage of
$\frac{|V(R)|}{|V(G_0)|}$ of each method is shown in
Figure \ref{fig.deg_V}. \LCTC still achieve the lowest percentage
among all three methods,  \LAZY still achieves the worst performance
on both network. All methods increase the percentage with the
decreased averaged degree of query node, due to the decreasing largest
trussness.}
}

\begin{figure}[t]
\vskip -0.4cm
\centering \mbox{
\subfigure[Query Time]{\includegraphics[width=0.37\linewidth,height=2.5cm]{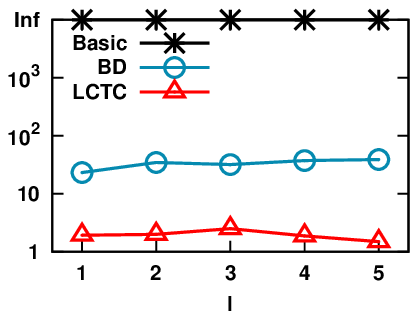}}\hskip -0.15in
\subfigure[The percentage]{\includegraphics[width=0.37\linewidth,height=2.5cm]{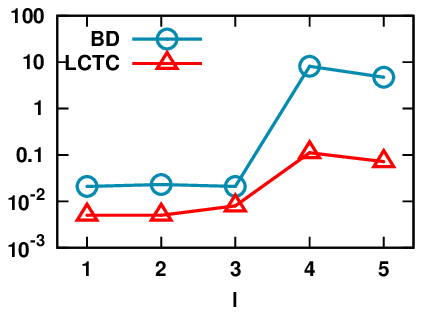}}
\hskip -0.15in
\subfigure[Density]{\includegraphics[width=0.37\linewidth,height=2.5cm]{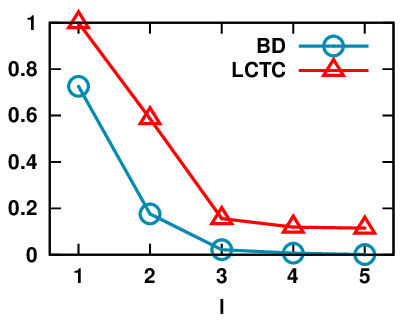}}
} \vskip -0.2in
\caption{DBLP: varying inner distance $l$} \label{fig.dist_time}
\vspace*{-0.3cm}
\end{figure}

\begin{figure}[t]
\vskip -0.2cm
\centering \mbox{
\subfigure[Query Time]{\includegraphics[width=0.37\linewidth,height=2.5cm]{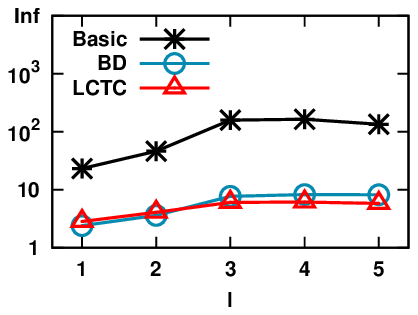}}\hskip -0.15in
\subfigure[The percentage]{\includegraphics[width=0.37\linewidth,height=2.5cm]{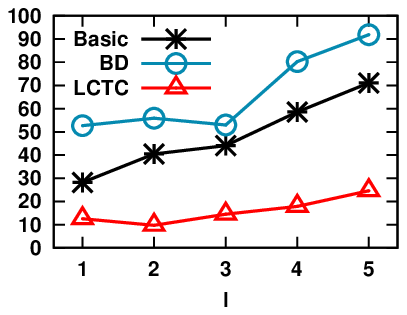}}
\hskip -0.15in
\subfigure[Density]{\includegraphics[width=0.37\linewidth,height=2.5cm]{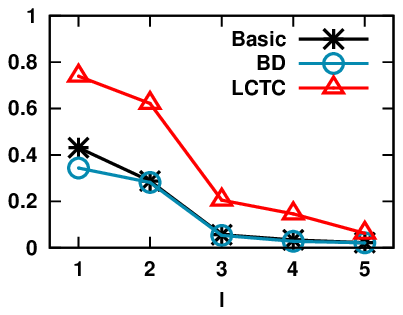}}
} \vskip -0.2in
\caption{Facebook: varying inner  distance $l$} \label{fig.dist_V}
\vspace*{-0.3cm}
\end{figure}

Third, we vary the inter-distance $l$ within query nodes from 1
to 5. For each $l$ value, we randomly select 100 sets of 3 query
nodes, in which the inter-distance of query nodes is to be $l$. \RV{ We
report the average runtime, the average percentage of
avoiding \FRE and density.} The results for DBLP and Facebook are shown in
Figure \ref{fig.dist_time} and Figure \ref{fig.dist_V}, respectively. \RV{The
performance in terms of runtime, the percentage of avoiding \FRE and density are
similar to the results observed.} All methods increase the percentage
while the inter-distance $l$ increases. This is because the diameter
of community increases, and therefore the less number of nodes can be
removed from graph.
\LCTC outperforms the others.

\comment{
We
report the averaged query time for each inter-distance group in
Figure \ref{fig.dist_time}. Obviously, \LCTC is most efficient,
and \CTC is worst. All methods achieve stable efficiency performance
by increasing the inter-distance, with a slightly increasing running
time, which is caused by the increased size of $G_0$. \LL{Moreover, we
report he averaged the percentage of $\frac{|V(R)|}{|V(G_0)|}$ of each
method by varying $l$ in Figure \ref{fig.dist_V}. First, all methods
increase the percentage with the increased inter-distance $l$, because
the diameter of community increases and the less nodes can be removed
from graph. \LCTC has the best performance again among all methods,
and \LAZY is worst.}
}


\begin{table}[h]
\begin{center}
\vspace*{-0.5cm}
\scriptsize
\caption[]{\textbf{Index size and index construction time}}\label{tab:index}
\begin{tabular}{|l|r|r|r|}
\hline
{\bf Network} & {\bf Graph Size (M)} & {\bf Index Size (M)} & {\bf Index Time (s)} \\ \hline
Facebook	& 	0.9	& 	1.3	& 7.4 \\ \hline
Amazon		& 12	& 	19	& 	6.7 \\ \hline
DBLP	& 	13		& 20		& 14 \\ \hline
Youtube	& 		37	& 	59		& 76 \\ \hline
LiveJounarl	& 	478		& 666		 	& 2,142 \\ \hline
Orkut		& 1,640	& 	2,190	& 	21,012 \\ \hline
\end{tabular}
\vspace*{-0.1cm}
\end{center}
\end{table}

We report the simple $k$-truss index in terms of index size
(Megabytes) and index construction time (seconds) in
Table \ref{tab:index}. The size of the $k$-truss index is 1.6 times of
the original graph size, which confirms that the simple $k$-truss
indexing scheme has $O(m)$ space complexity and is very compact. The
index construction is very efficient.
%
%

\eat{

\subsection{Efficient Truss Maintenance}

In this section, we evaluate the improvement performance of efficient
truss maintenance techniques. We consider a variant method of \lazy
equipped with the basic truss maintenance in
Algorithm \ref{algo:simple_trssmt}, and compare it with the
method \lazy equipped with fast truss maintenance in
Algorithm \ref{algo:trssmt}. We report the averaged query time for two
methods varying different query degree groups on DBLP network in
Figure \ref{fig.truss_maintain}. As we can see, the efficient truss
maintenance method takes less time than the basic one. Indeed, the
efficient truss maintenance method improves the query processing
performance, especially when the averaged degree of query nodes is low
and the corresponding graph $G_0$ is large.

\begin{figure}[h]
\vskip -0.1cm
\centering \mbox{
\subfigure[DBLP]{\includegraphics[width=0.5\linewidth,height=2.5cm]{./Figure/exp/DBLP/fast_time.eps}}
} \vskip -0.2in
\caption{Query time (in seconds) of \lazy methods respectively with basic and fast truss maintenance} \label{fig.truss_maintain}
\end{figure}
}


\noindent
{\bf \JY{Exp-2} A Case Study on DBLP}: We construct a collaboration
network from the raw DBLP data
set\footnote{\small{\url{http://dblp.uni-trier.de/xml/}}} for a case
study. A vertex represents an author, and an edge between two authors
indicates they have co-authored no less than 3 times. This DBLP graph
contains 234,879 vertices and 541,814 edges.

\RV{
We use the query $Q=$ \{``Alon Y. Halevy", ``Michael J. Franklin", ``Jeffrey D. Ullman", ``Jennifer Widom"\} to test our \ctc model for detecting the community. Figure \ref{fig.sensetive}(a) shows
$G_0$ that is the maximal connected $9$-truss containing $Q$.
This entire graph has 73 nodes, 486 edges, edge density of 0.18 and diameter of 4.
As we can see that most black nodes span long distance to reach at each other. They are loosely connected with query nodes by some midsts.  
Our method \LCTC removes these balck nodes and finds a \ctc for $Q$ shown in
Figure \ref{fig.sensetive}(b), which is a 9-truss of diameter 2. It
has 14 authors, 81 edges and the edge density of 0.89.  The community
does not inclue any authors in 9-truss, and thoes other are far away
from queried authors.
}


\begin{figure}[t]
\scriptsize
 \vspace{-0.2cm}
\centering
{
\subfigure[$G_0$]{\includegraphics[width=0.45\linewidth]{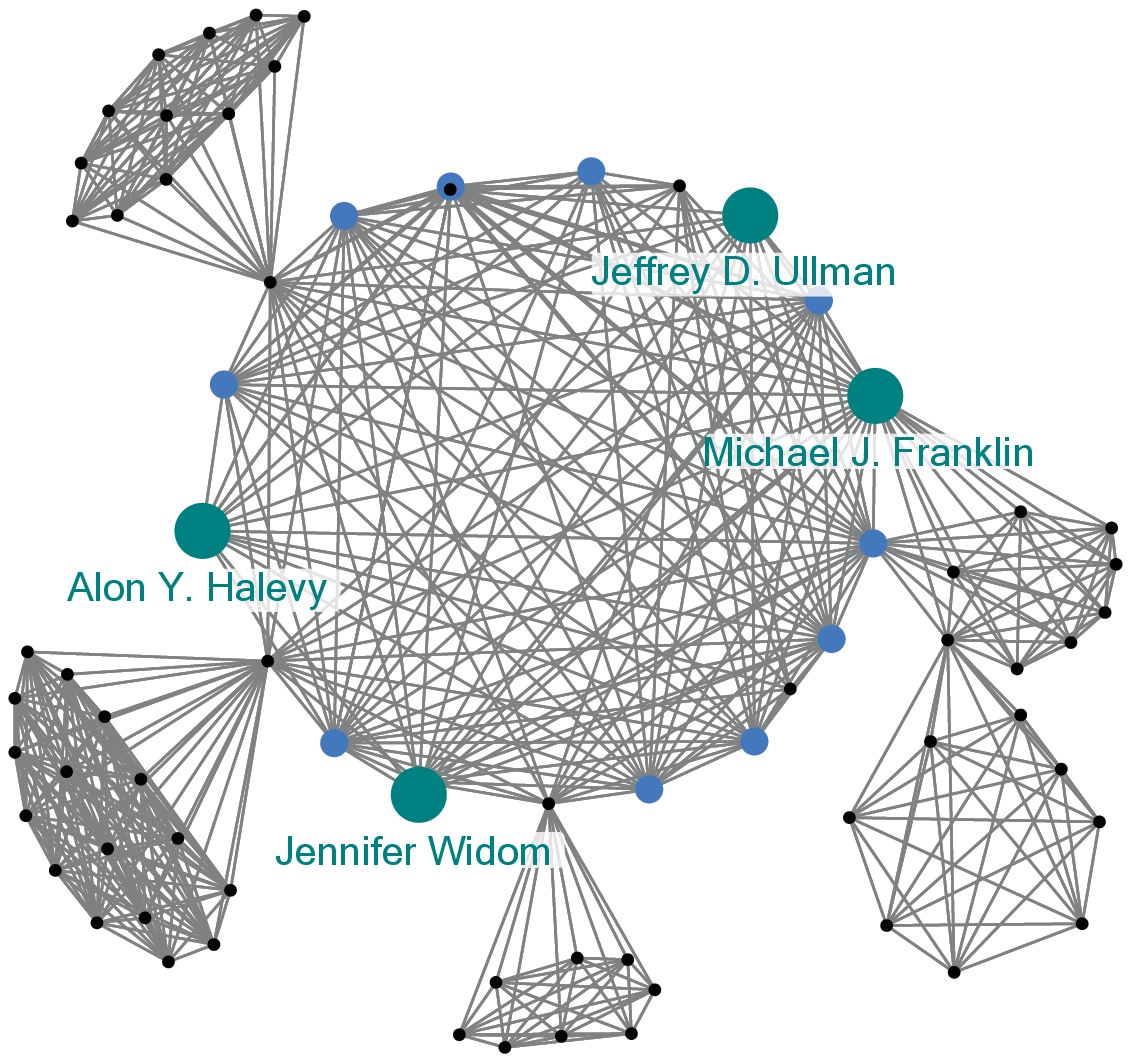}}
\subfigure[\LCTC]{\includegraphics[width=0.45\linewidth]{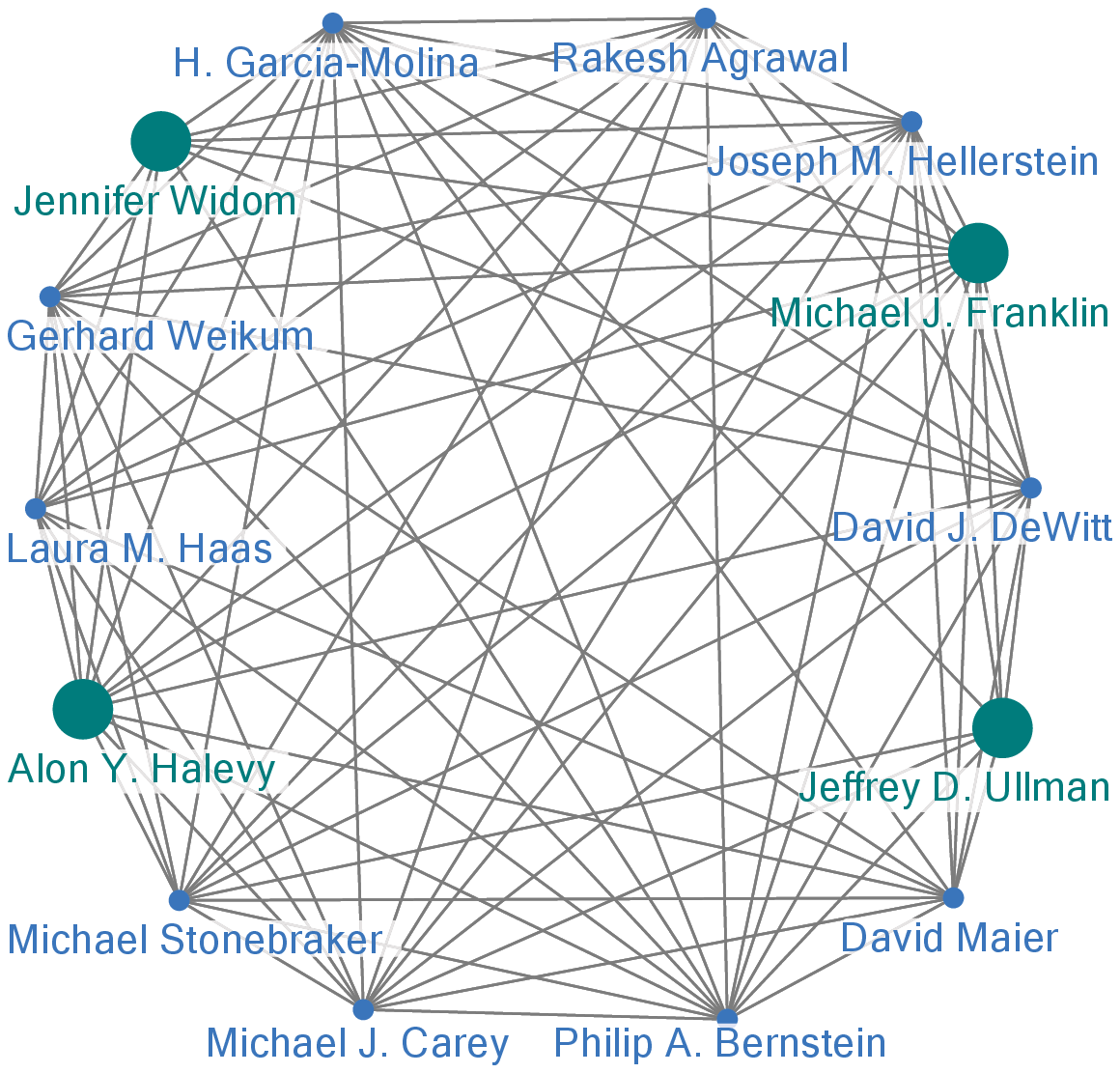}}
}
\vspace{-0.4cm}
\caption{Community search on DBLP network using query $Q=$\{``Alon Y. Halevy", ``Michael J. Franklin", ``Jeffrey D. Ullman", ``Jennifer Widom"\} }
\label{fig.sensetive}\vspace{-0.2cm}
\end{figure}

\begin{figure}[t]
\vskip -0.4cm
\centering \mbox{\hskip -0.1in
\subfigure[$F_1$ score]{\includegraphics[width=0.38\linewidth]{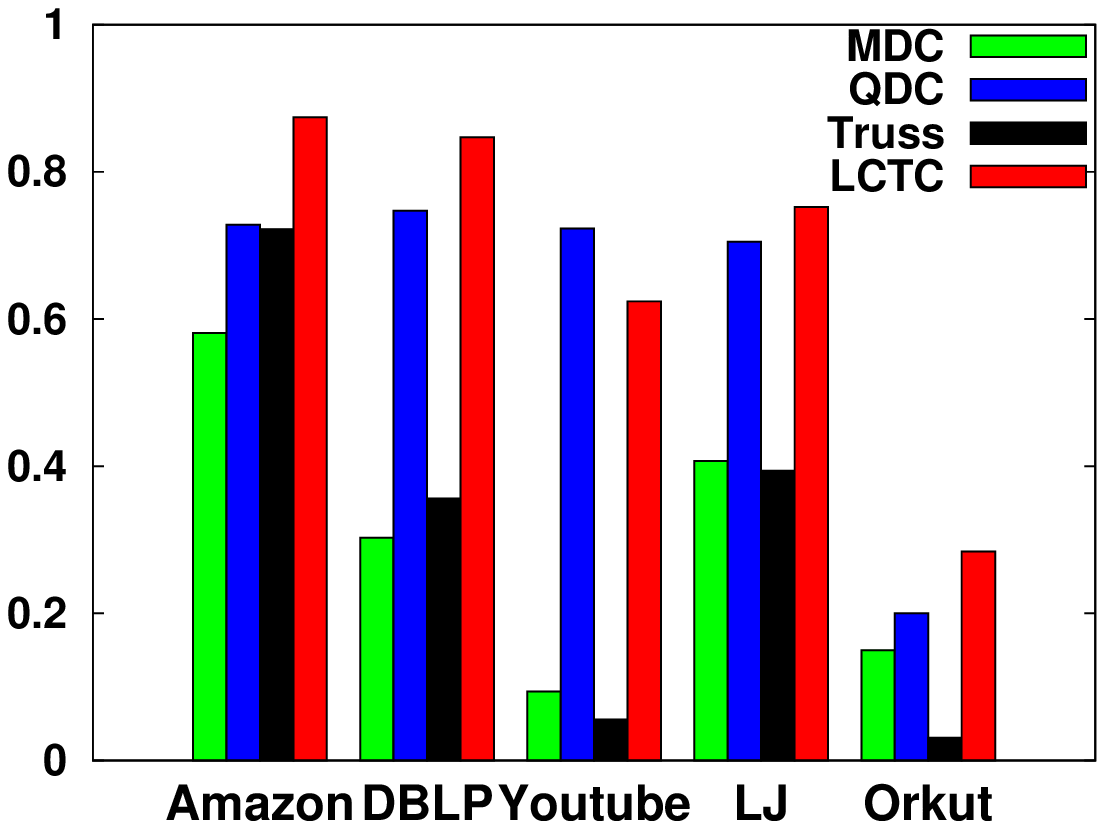}} \hskip -0.1in
\subfigure[Query
time]{\includegraphics[width=0.38\linewidth]{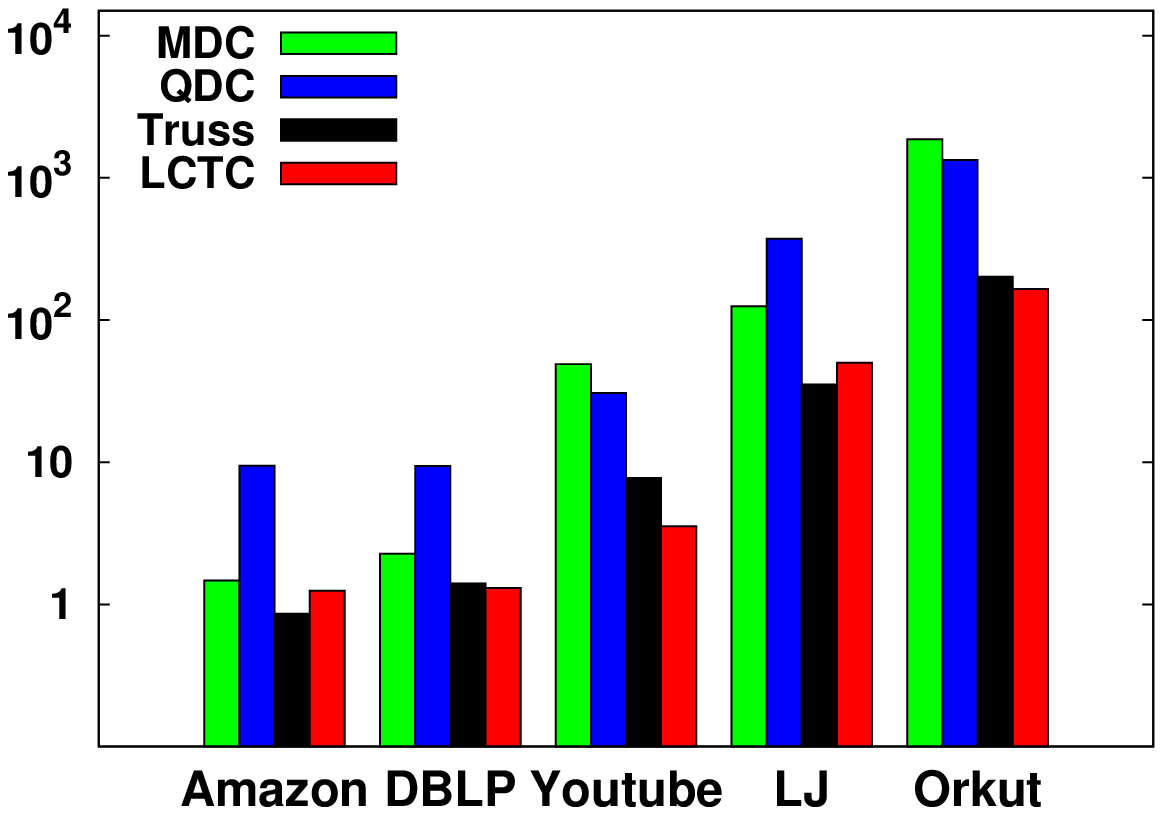}}
\hskip -0.1in
\subfigure[Reduction]{\includegraphics[width=0.38\linewidth]{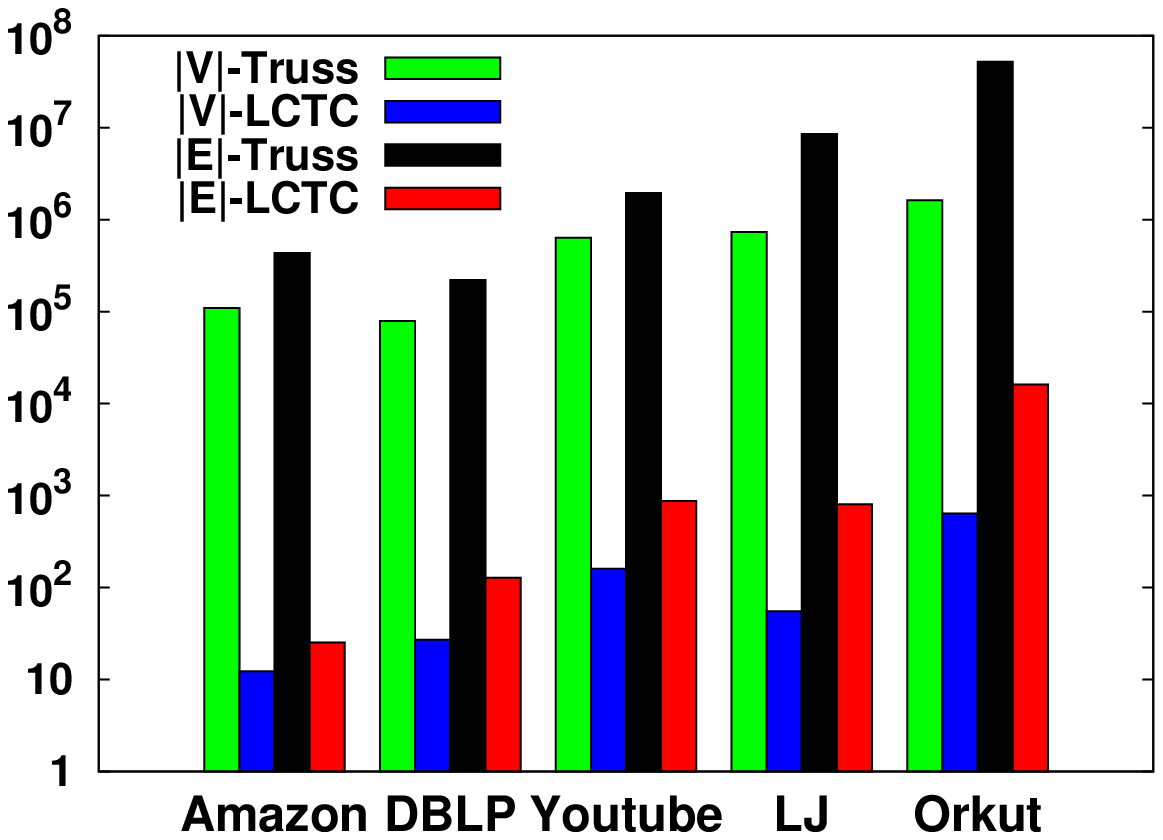}}
}
\vskip -0.2in
\caption{Quality evaluation on networks with ground-truth communities} \label{fig.quality}
\vspace*{-0.4cm}
\end{figure}

\noindent
{\bf Exp-3 The Quality \JY{by Ground-Truth}}: To evaluate
the effectiveness of different community models, we compare \LCTC with
three other methods \MDC, \QDC and \TRUS using the 5 networks, DBLP,
Amazon, Youtube, LiveJournal, and Orkut, with ground-truth
communities \cite{YangL12}.
%
%
We randomly select query nodes that appear in a unique ground-truth
community, and select 1,000 sets of such query nodes with the size
randomly ranging from 1 to 16.
%
%
%
%
We evaluate the accuracy by the F1-score of the detected community,
and report the averaged F1-score over all query cases.

Figure \ref{fig.quality}(a) shows the F1-score.  Our method achieves the
highest F1-score on most networks. \QDC has the second best
performance, which outperforms \LCTC on Youtube network. \MDC does not
perform well due to the fixed distance and size constraints.
\RV{We observe that the accuracy drops on Orkut for most methods. One possible reason is that many ground-truth communities in Orkut are not densely connected, which violates the assumption of all dense community models. Another reason is that the community membership per node on Orkut is much larger than that on other networks \cite{YangL12}. The large overlap of ground-truth communities makes them difficult to be detected accurately.}
Figure \ref{fig.quality}(b) shows that \LCTC runs much faster than \MDC and \QDC, and is close to \TRUS.
%
%
Figure \ref{fig.quality}(c) shows the size of communities detected by
\LCTC and \TRUS, in terms of the number of vertices and edges.
%
%
As we can see, the number of nodes ($|V|$-) and the number of edges
($|E|$-) in communities detected by \LCTC are much less than those by
\TRUS on all networks. It confirms the power of eliminating irrelevant
nodes from discovered communities by our \LCTC.

\begin{figure}[t]
\vskip -0.2cm
\centering \mbox{
\subfigure[Diameter]{\includegraphics[width=0.5\linewidth,height=2.5cm]{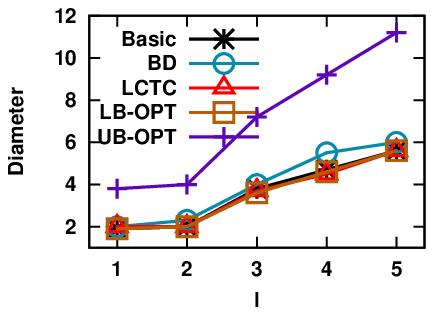}}
\subfigure[Trussness]{\includegraphics[width=0.5\linewidth,height=2.5cm]{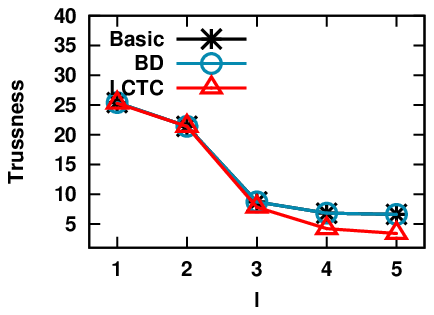}}
} \vskip -0.2in
\caption{Varying the inner distance $l$ on Facebook} \label{fig.diam_truss}
\vspace*{-0.2cm}
\end{figure}


\noindent
{\bf \JY{Exp-4} Diameter and Trussness Approximation}: We evaluate the
diameter approximation of detected communities by our methods on
Facebook network.
%
%
Here, we take the lower bound of the optimal diameter \JY{(LB-OPT)} as
the smallest query distance $\dist_{R}(R,Q)$, where $R$ is the
community detected by method \CTC. We show the curve of
$2\dist_{R}(R,Q)$, which serves as the upper bound of smallest
diameter \JY{(UB-OPT)} by Lemma \ref{lemma.disdia}. The averaged
diameters of communities detected by different methods are reported in
Figure \ref{fig.diam_truss}(a), where we vary the inter-distance $l$.
The diameters of detected communities obtained by all our methods are
very close to the lower bound of optimal one.
Figure \ref{fig.diam_truss}(b) shows the maximum trussness of detected
communities by our methods. \CTC and \LAZY globally
search the $k$-truss containing query nodes on the entire graph, and
the detected communities have the maximum trussness $k$.
\JY{
\LCTC can
detect the trussness of communities
which are very close to \CTC and \LAZY, by
searching over a small graph locally.
}
\LCTC
balances the efficiency and effectiveness well.

\eat{
\noindent
{\bf \JY{Exp-5} Varying Maximum Trussness $k$}: In this experiment, we
evalute our three methods that do not find the truss community with
the real maximum trussness, but with a given maximum value $k$. The
parameter $k$ ranges from 2 to 8.  For comparison, we also report the real
 maximum trussness case as ``max'' in the horizontal axis. The results of diameter
varying by $k$ are respectively shown in Figure \ref{fig.fixk} (a) and
(b).
In Figure \ref{fig.fixk} (a), when the maximum trussness $k$
decreases, the lower bound of optimal diameter also decreases as
expected, but the margin is small. On the other hand, even though
trussness $k$ greatly decreases from 8 to 2, the communities detected
by our methods are nearly and very close to the lower bound of optimal
one. \LCTC and \CTC ahieve the smallest diameter for $k=4$, and \LAZY
does it for $k=6$. This indicates that it is hard to choose the best
$k$ for smallest diameter in real applications. In addition, as we can
see in Figure \ref{fig.fixk} (b), the density become worse by \CTC
and \LAZY methods when $k$ decreases, which shows the necessary of
maximum trussness constraint. On the other hand, \LCTC achieve stable
edge density when parameter $k$ varies.
}

\begin{figure}[t]
\vskip -0.25cm
\centering
\includegraphics[width=0.5\linewidth,height=2.5cm]{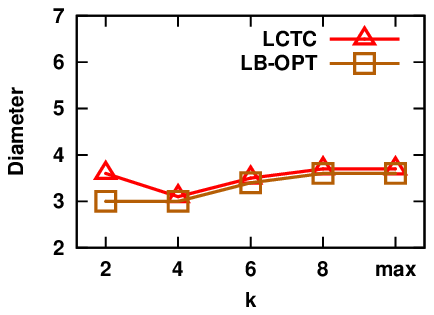}
\vskip -0.2in
\caption{Diameter v.s. the maximum trussness $k$ on Facebook} \label{fig.fixk}
\end{figure}

\RV{
\noindent
{\bf \JY{Exp-5} Varying Maximum Trussness $k$}: 
In this experiment, we
evalute our method \LCTC that do not find the truss community with
the real maximum trussness, but with a given maximum value $k$. We test different $k$ ranged from 2 to to ``max ", which is the real largest trussness could be.  The diameter of found community by \LCTC is reported in Figure \ref{fig.fixk}. With $k$
decreases, the lower bound of optimal diameter also decreases from 3.6 to 3.0,  but the margin is small. Meanwhile, the communities detected by \LCTC are very close to the optimal
one for any $k$. The approximation ratio is not greater than 1.2. This indicates that our model with the maximum trussness constraint have the adavantage of parameter-free. }

\noindent
{\bf Exp-6 Varying \LCTC parameters}: In this experiment, we test the performance of \LCTC by varying  parameters $\eta$ and $\gamma$. We used the same query nodes that are selected in Exp-3 on DBLP network. The similar results can be also observed on other 4 networks in this paper. $\eta=1000$ and $\gamma=3$ is the default setting for \LCTC.
For the parameter $\eta$, we firstly vary it from 100 to 2000. The results of F1-score, the number of community vertices $|V|$ and the running time. are reported in Figure \ref{fig.lctc_eta}. As we can see, the number of community vertices increases when $\eta$ increases from 100 to 500, and then keeps stable for larger $\eta$. It shows that the default setting $\eta=1000$ is large enough. Moreover, \LCTC achieves the stable performance of F1-score and running time by varying $\eta$. We also test the parameter $\gamma$, and report the results on Figure \ref{fig.lctc_gamma}. The number of community vertices increases with the increased $\gamma$. Because \LCTC with a larger $\gamma$ can detected the community of a larger trussness, and the number of vertices to be removed is reduced. On the other hand, the F1-score increases with increasing $\gamma$ at first, but it drop slightly when $\gamma$ further increases. The running time of \LCTC keeps table.

\begin{figure}[t]
\centering \mbox{
\subfigure[$|V|-$\LCTC]{\includegraphics[width=0.35\linewidth,height=2.5cm]{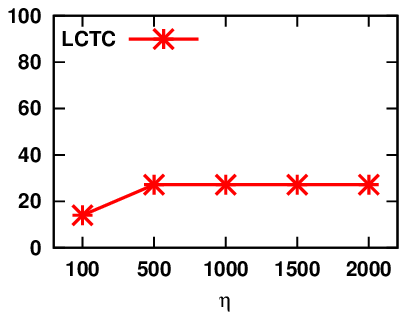}}\hskip -0.15in
\subfigure[F1-Score]{\includegraphics[width=0.35\linewidth,height=2.5cm]{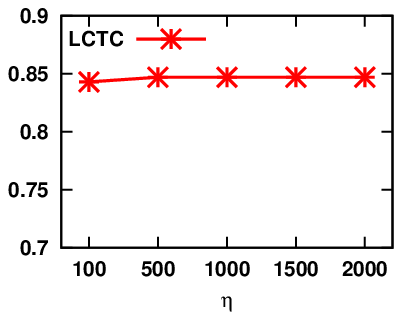}} \hskip -0.15in
\subfigure[Query Time]{\includegraphics[width=0.35\linewidth,height=2.5cm]{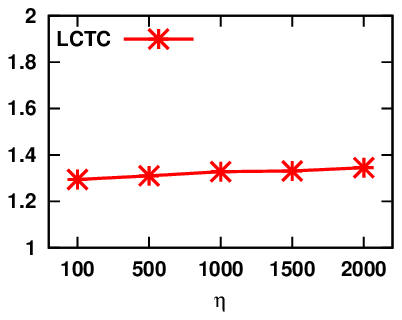}}
} \vskip -0.2in
\caption{DBLP: varying parameter $\eta$ of \LCTC} \label{fig.lctc_eta}
\vspace*{-0.3cm}
\end{figure}

\begin{figure}[t]
\centering \mbox{
\subfigure[$|V|-$\LCTC]{\includegraphics[width=0.35\linewidth,height=2.5cm]{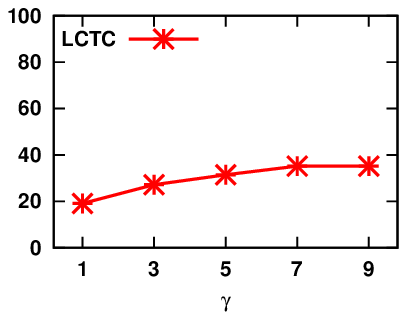}}\hskip -0.15in
\subfigure[F1-Score]{\includegraphics[width=0.35\linewidth,height=2.5cm]{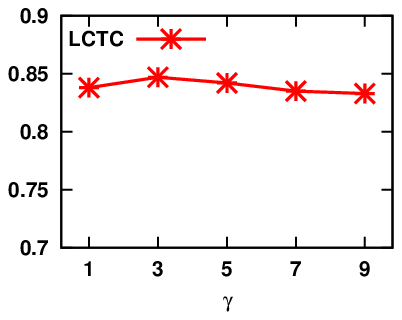}} \hskip -0.15in
\subfigure[Query Time]{\includegraphics[width=0.35\linewidth,height=2.5cm]{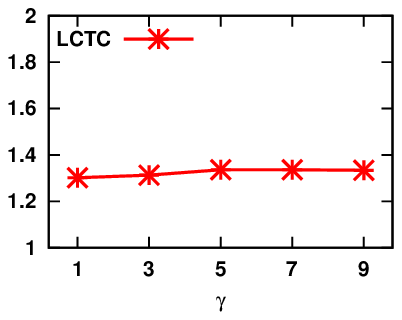}}
} \vskip -0.2in
\caption{DBLP: varying parameter $\gamma$ of \LCTC} \label{fig.lctc_gamma}
\vspace*{-0.3cm}
\end{figure}

%% file: relate.tex
\RV{

In this section, we firstly discuss the rationale of our designed model, and then review the most related work to our study, which contains community search, community detection, and dense subgraph mining.

\subsection{Design Decisions} \label{sec.design}

Here, we discuss several natural candidates for community models and provide a rationale for our definition of \ctc.

\laks{
\stitle{Diameter vs query distance.} Being closely related to the query nodes is a natural desirable property for nodes to be included in a community. In the literature, small diameter has been regularly considered as an important hallmark of a good community -- see e.g., \cite{Edachery99graphclustering, doddi2000approximation, gonzalez1985clustering, edachery1999graph}. Thus, minimizing diameter in identifying communities has a natural motivation.

Secondly, by definition, a community with a small diameter will also have small query distance from its nodes. On the other hand, minimizing query distance ignores the distance between non-query nodes in the community. In this sense, small diameter is a strictly stronger property than small query distance. Example \ref{ex-qd-dia} illustrates this point and the value of minimizing diameter as opposed to just query distance.

}

\laks{
\stitle{Trading trussness for diameter.} Every $k$-truss is also a $(k-1)$-truss by definition. Thus, relaxing the maximum trussness requirement may allow us to find a community with a smaller diameter by sacrificing trussness.  One problem is that the variation of diameter as trussness decreases, may not be smooth but may face a sudden drop as trussness decreases to a low value. E.g., continuing with the example of Figure \ref{fig.community}(a) with query $Q=\{q_1, q_2, q_3\}$, our CTC model yields a community with the highest trussness $k=4$ and diameter 3, as in Figure \ref{fig.community}(b). When $k=3$, the 3-truss containing $Q$ with the smallest diameter is still Figure \ref{fig.community}(b). However, when $k=2$, the cycle of $\{(q_1, t), $ $(t, q_3),$ $ (q_3, v_4),$ $ (v_4, q_2),$ $ (q_2, q_1)\}$ turns out to be the 2-truss containing $Q$ and its diameter is 2. However, this is loosely connected and has a low edge density. In general, for a small $k$, the $k$-truss community found by removing free riders may have loosely connected structure and thus may be noisy. One advantage of our approach is that it is parameter-free. However, if a user would like to explore trading trussness for diameter, it is straightforward to extend our algorithms (Algorithms \ref{algo:mindia} and \ref{algo:findg0}) to treat the desired trussness $k$ as a constraint instead of maximizing trussness. Finally, another way of combining trussness and diameter is using a weighted combination, but this comes with the challenge of tuning the weights. Our parameter-free approach of minimizing diameter while keeping trussness at the  maximum value is a reasonable choice.
}

\laks{
\stitle{Constraining community size.}
At first, it appears that we can minimize or avoid free riders by bounding the size of a community. However, sizes of commuities may vary widely and it is difficult, if possible at all, to impose proper bounds on acceptable community sizes. Moreover, bounding the size of the community may render the problem of finding a query driven community inapproximable w.r.t. any factor. Specifically, consider the special case of finding a $k$-truss of size at most a given parameter $\ell$ that contains $Q=\emptyset$. This subsumes the $k$-clique problem, which is not approximable within any reasonable factor \cite{haastad1996clique}. By contrast, minimizing diameter instead of size admits efficient approximation. Indeed, our formulation does address community size indirectly. The larger the $k$, the smaller the size of the $k$-truss. Our CTC model maximizes trussness. Furthermore, by minimizing the diameter, it helps remove free riders, thus reducing the size in a disciplined manner. On the algorithmic side, our LCTC method (Section \ref{sec.local}) actually uses a size threshold to prune the search space and improve efficiency. Thus, LCTC controls the size of a community in a heuristic manner.
}

}


\subsection{Community Search}
Recently, several community search models have been studied, including $k$-truss  \cite{huang2014}, quasi-clique  \cite{CuiXWLW13}, $k$-core  \cite{sozio2010, cui2014local}, influential community  \cite{li2015influential} and  query biased densest subgraph  \cite{wu2015robust}. Here, we compare these models with our proposed \ctc model w.r.t.\ three aspects: (i) consideration of query nodes, (ii) cohesive structure, and (iii) quality approximation. 

\stitle{Query nodes.}
Cui et al.\ \cite{CuiXWLW13} have recently studied the problem of online search of overlapping communities for a query node by designing a new $\alpha$-adjacency $\gamma$-quasi-$k$-clique model. Huang et al.\  \cite{huang2014} propose a $k$-truss community model based on triangle adjacency, to find all overlapping communities of a query node. They ignore the diameter of the resulting community.  Cui et al.\ \cite{cui2014local} find a $k$-core community for a query node using local search.  In addition, influential community model \cite{li2015influential} finds top-$r$ communities with the highest influence scores over the entire graph; no query nodes are considered.  Extending any of above models from one (or zero) query node to multiple query nodes raises new challenges. First, for the models with one query node and a parameter $k$, the search algorithm can easily start from this node to find qualified subgraphs. For multiple nodes, it is non-trivial for the search algorithm to determine the start point and search directions, which can quickly connect all query nodes. Second, for a given parameter $k$, the connected dense subgraph containing all query nodes may not exist. Thus, it requires the search algorithm to automatically determine the proper $k$ for different query nodes.

\eat{
There exist several recent works on community search  for multiple query node(s). 
Sozio et al.\ \cite{sozio2010} proposed a $k$-core based community model with distance and size restriction. 
Most recently, Wu et al. \cite{wu2015robust} studied the query biased densest connected subgraph (QDC) problem for avoiding subgraphs irrelevant to query nodes in the community found. This is the most related work to ours,  we compare our model with them below.
\note[Laks]{Isn't THIS the section where we should compare all of them?}
\note[Xin]{I compare all methods in this subsection. We can remove them to another subsection or a single section. I am fine with both.}
}

\stitle{Cohesive structure.} \cite{sozio2010} and \cite{wu2015robust} support community search of multiple query nodes similarly to us, thus they are most related to our work.  Sozio et al.\ \cite{sozio2010} proposed a $k$-core based community model, called Cocktail Party model, with distance and size constraints.
\RV{Our proposed \ctc model is based on connected $k$-truss. Conceptually, $k$-truss is a more cohesive definition than $k$-core, as $k$-truss is based on triangles whereas $k$-core simply considers node degree \cite{WangC12}.
}
\eat{
Conceptually, $k$-truss is a more cohesive definition than $k$-core, as $k$-truss is based on triangles whereas $k$-core simply considers node degree \cite{WangC12}. Our proposed \ctc model is based on connected $k$-truss, and the Cocktail Party model is based on connected $k$-core.
}
\eat{Conceptually, $k$-truss is a more cohesive definition than $k$-core, as $k$-truss is based on triangles whereas $k$-core simply considers node degree \cite{WangC12}. }
 Most recently, Wu et al.\ \cite{wu2015robust} studied the query biased densest connected subgraph (QDC) problem for avoiding subgraphs irrelevant to query nodes in the community found. \LL{While QDC \cite{wu2015robust} is also defined based on a connected graph containing $Q$ similarly to CTC, it optimizes a fundamentally different function called query biased edge density, which is calculated as the overall edge weight averaged over the weight of nodes in a community.}

\note[Laks]{OK, this discussion shows that our model is somewhat different from previous ones. But we also need to argue in what way our model is more appropriate as a model for community, compared to previous models.} \note[xin]{Our model not only consider the community nodes should be relevant(close), but also require cohesive structure of each edges(node) in the community. The found community has the quality guarantee. }
\note[Laks]{You make interesting points above. Can you give an example to illustrate the last point?} \note[Xin]{
I will try to show a experiment comparison later.}

\stitle{Quality approximation.} Both problems proposed in \cite{sozio2010} and \cite{wu2015robust} are NP-hard to compute, and do not admit approximations without further assumptions.  \cite{wu2015robust} gives an approximation solution of QDC by relaxing the problem. Unfortunately, as the authors show themselves \cite{wu2015robust}, this could fail in real applications, for two reasons. First, the algorithm may find a solution consisting of several connected components with query nodes split between them. Second, the approximation factor can be large, which can deteriorate further with a larger number of query nodes.
\eat{
based on an assumption that the result of one relaxation problem $QDC'$ exists a connected subgraph containing $Q$ and at least one non-query node, which could fail in real applications as shown in \cite{wu2015robust}.  \note[Laks]{How? Can you explain?} \note[Xin]{Please ref to first paragraph in Section 6.3 of \cite{wu2015robust}}.
 The reason has two folds. First, the algorithm finds the optimal solution consisted of several connected components with the largest edge density. But, the query nodes are fallen into different components. The other reason is that not any non-query nodes are in the community. Moreover, the approximate ratio could be large for multiple query nodes. The ratio is $\pi(T)/(\pi(T)-\pi(Q))$ (Lemma 12 of \cite{wu2015robust}), where $T$ is the nodes set of answered community, and $\pi(T)$ represents the total of query biased node weight of $T$.  $\pi(T)-\pi(Q)$ equals to the node weight of non-query nodes, which can be small. With more query nodes, $\pi(T)$ becomes larger, and the approximation ratio may be larger.
\note[Laks]{What is $\pi$? Let's also explain in words what we are saying here. This is not a technical section of the paper.} \note[Xin]{See above.}
}
In contrast, we provide an efficient 2-approximation algorithm for finding the closest truss community containing any set of query nodes. We provide a heuristic algorithm based on local exploration which significantly improves the efficiency and show that on several real networks, it delivers a high-quality solution.

\subsection{Community Detection}
\eat{
The goal of community detection is to identify all communities in the entire network. A typical method for finding communities is to optimize the modularity measure \cite{newman2004fast}. Generally, community detection falls into two major categories: non-overlapping \cite{Rosvall08, yang2012clustering} and overlapping community detection \cite{Palla05, Ahn10}. All these methods consider  static communities, where the networks are partitioned a priori. Query nodes are not considered since their focus is not community search. 
As such, these works on community detection are significantly different from our goal of query driven community search.
}

The goal of community detection is to identify all communities in the entire network. A typical method for finding communities is to optimize the modularity measure \cite{newman2004fast}. Generally, community detection falls into two major categories: non-overlapping \cite{Newman04, Rosvall08, yang2012clustering} and overlapping community detection \cite{Palla05, Ahn10, XieKS13, yang2013overlapping}. All these methods consider  static communities, where the networks are partitioned a priori. Query nodes are not considered since their focus is not community search. \cite{lancichinetti2009community} surveys several community detection methods and evaluates their performance using rigorous tests. \cite{xie2013labelrankt} proposes an online distributed algorithm for community detection in dynamic networks using label propagation. As such, these works on community detection are significantly different from our goal of query driven community search.

\subsection{Dense Subgraph Mining}
\eat{There is a very large body of work on mining dense subgraph patterns, including clique \cite{BronK73}, quasi-clique \cite{Tsour13}, $k$-core \cite{BatageljZ03, JCheng11}, $k$-truss \cite{cohen2008, WangC12}
 to name a few. Various studies have been done on core decomposition and truss decomposition in different settings, including in-memory algorithms \cite{BatageljZ03, cohen2008}, external-memory algorithms \cite{JCheng11, WangC12}, and MapReduce \cite{Cohen09}.
None of these works considers query nodes, which as we have discussed earlier, raise major computational challenges.}

There is a very large body of work on mining dense subgraph patterns, including clique \cite{BronK73, JCheng10, JWang13, JXiang13}, quasi-clique \cite{Tsour13}, $k$-core \cite{BatageljZ03, JCheng11}, $k$-truss \cite{cohen2008, WangC12, YZhang12}, dense neighborhood graph \cite{NWang10},
 to name a few.

Clique and quasi-clique enumeration methods include the classical algorithm \cite{BronK73}, the external-memory $H^*$-graph algorithm \cite{JCheng10}, redundancy-aware clique enumeration \cite{JWang13}, maximum clique computation using MapReduce \cite{JXiang13}, and optimal quasi-clique mining \cite{Tsour13}. Various studies have been done on core decomposition and truss decomposition in different settings, including in-memory algorithms \cite{BatageljZ03, cohen2008, YZhang12}, external-memory algorithms \cite{JCheng11, WangC12}, and MapReduce  \cite{Cohen09}.  \cite{huang2014, YZhang12} designed an incremental algorithm for updating a $k$-truss with edge insertions/deletions. Wang et al.\ \cite{NWang10} studied a dense neighborhood graph based on common neighbors.
None of these works considers query nodes, which as we have discussed earlier, raise major computational challenges.